\documentclass[10pt,journal,compsoc]{IEEEtran}

\ifCLASSINFOpdf
\else
\fi

\usepackage{times,amsmath,epsfig}
\usepackage{epstopdf}
\usepackage{color}
\usepackage{graphicx}  
\usepackage{cases}
\usepackage{algorithm}
\usepackage{algorithmic}  
\usepackage{textcomp} 
\usepackage{tikz}
\usepackage{amsthm}
\usepackage{arydshln}
\usepackage{stackengine}
\usepackage{physics}

\usepackage[utf8]{inputenc}
\usepackage[english]{babel}
\newtheorem{proposition}{Proposition}
\newtheorem{definition}{Definition}
\newtheorem{theorem}{Theorem}

\theoremstyle{plain}
  
\makeatletter
  \newcommand*{\rom}[1]{\expandafter\@slowromancap\romannumeral #1@}
\makeatother

\usepackage{ragged2e}
\usepackage{enumerate}
\usepackage{cite}
\usepackage{amsmath,amssymb,amsfonts}
\usepackage{algorithmic}
\usepackage{graphicx}
\usepackage{textcomp}
\usepackage{color}
\usepackage[normalem]{ulem}
\useunder{\uline}{\ul}{}
\def\BibTeX{{\rm B\kern-.05em{\sc i\kern-.025em b}\kern-.08em
    T\kern-.1667em\lower.7ex\hbox{E}\kern-.125emX}}

\usepackage{multirow}

 \ifCLASSOPTIONcompsoc 
 \usepackage[caption=false,font=normalsize,labelfont=sf,textfont=sf]{subfig}
 \else
 \usepackage[caption=false,font=footnotesize]{subfig}
 \fi
 
\def\0{{\mathbf 0}}
\def\1{{\mathbf 1}}

\def\b{{\mathbf b}}
\def\c{{\mathbf c}}

\def\e{{\mathbf e}}

\def\n{{\mathbf n}}
\def\p{{\mathbf p}}
\def\q{{\mathbf q}}

\def\u{{\mathbf u}}
\def\v{{\mathbf v}}
\def\x{{\mathbf x}}

\def\A{{\mathbf A}}
\def\B{{\mathbf B}}
\def\C{{\mathbf C}}
\def\D{{\mathbf D}}

\def\F{{\mathbf F}}

\def\H{{\mathbf H}}
\def\I{{\mathbf I}}

\def\L{{\mathbf L}}
\def\M{{\mathbf M}}

\def\S{{\mathbf S}}

\def\W{{\mathbf W}}

\def\cC{{\mathcal C}}
\def\cE{{\mathcal E}}
\def\cF{{\mathcal F}}
\def\cG{{\mathcal G}}
\def\cH{{\mathcal H}}

\def\cK{{\mathcal K}}
\def\cL{{\mathcal L}}

\def\cS{{\mathcal S}}

\def\cU{{\mathcal U}}
\def\cV{{\mathcal V}}

\def\ie{{\textit{i.e.}}}
\def\eg{{\textit{e.g.}}}

\begin{document}
%
\title{Point Cloud Sampling via Graph Balancing and Gershgorin Disc Alignment}
%
%
%

\author{Chinthaka~Dinesh,~\IEEEmembership{Student~Member,~IEEE,}
        Gene~Cheung,~\IEEEmembership{Fellow,~IEEE,}
        and~Ivan~V.~Baji\'{c},~\IEEEmembership{Senior~Member,~IEEE}
\IEEEcompsocitemizethanks{\IEEEcompsocthanksitem Chinthaka Dinesh and Ivan V. Baji\'{c} are with the School of Engineering
Science, Simon Fraser University, Burnaby, BC, Canada. \protect\\ E-mail: hchintha@sfu.ca and ibajic@ensc.sfu.ca.
\IEEEcompsocthanksitem Gene Cheung is with the Department of Electrical Engineering \& Computer Science, York University, Toronto, Canada. E-mail: genec@yorku.ca. }
}

%
%

\markboth{}%
{Shell \MakeLowercase{\textit{et al.}}: Bare Demo of IEEEtran.cls for Computer Society Journals}
%




\IEEEtitleabstractindextext{
\begin{abstract}
\justifying{3D point cloud (PC)---a collection of discrete geometric samples of a physical object's surface---is typically large in size, which entails expensive subsequent operations like viewpoint image rendering and object recognition.
Leveraging on recent advances in graph sampling, we propose a fast PC sub-sampling algorithm that reduces its size while preserving the overall object shape.
Specifically, to articulate a sampling objective, we first assume a super-resolution (SR) method based on feature graph Laplacian regularization (FGLR) that reconstructs the original high-resolution PC, given 3D points chosen by a sampling matrix $\H$.
We prove that minimizing a worst-case SR reconstruction error is equivalent to maximizing the smallest eigenvalue $\lambda_{\min}$ of a matrix $\H^{\top} \H + \mu \cL$, where $\cL$ is a symmetric, positive semi-definite matrix computed from the neighborhood graph connecting the 3D points.
Instead, for fast computation we maximize a lower bound $\lambda^-_{\min}(\H^{\top} \H + \mu \cL)$ via selection of $\H$ in three steps.
Interpreting $\cL$ as a generalized graph Laplacian matrix corresponding to an unbalanced signed graph $\cG$, we first approximate $\cG$ with a balanced graph $\cG_B$ with the corresponding generalized graph Laplacian matrix $\cL_B$.
Second, leveraging on a recent theorem called Gershgorin disc perfect alignment (GDPA), we perform a similarity transform 
$\cL_p = \S \cL_B \S^{-1}$ so that Gershgorin disc left-ends of $\cL_p$ are all aligned at the same value $\lambda_{\min}(\cL_B)$.
Finally, we perform PC sub-sampling on $\cG_B$ using a graph sampling algorithm to maximize $\lambda^-_{\min}(\H^{\top} \H + \mu \cL_p)$ in roughly linear time. 
Experimental results show that 3D points chosen by our algorithm outperformed competing schemes both numerically and visually in SR reconstruction quality.}
\end{abstract}

}

\maketitle

%
\IEEEpeerreviewmaketitle

\IEEEraisesectionheading{
\section{Introduction}
\label{sec:introduction}}

\IEEEPARstart{P}{oint} cloud (PC) is a collection of typically non-uniform discrete geometric samples (\ie, 3D coordinates) of a physical object in 3D space. 
PC has become a popular 3D object representation format, for applications such as immersive communication and virtual / augmented reality (VR / AR)~\cite{shen2013, mekuria2017}.
A common PC can be very large in size---up to millions of 3D points. 
This means that subsequent 
operations like perspective rendering and object detection/recognition~\cite{zhou2018voxelnet} can be very computation-intensive.

To reduce computation load, one can perform \textit{PC sub-sampling}: selecting a representative 3D point subset while preserving the geometric structure of the original PC.
Early works in PC sub-sampling employed simple schemes like random or regular sampling that do not preserve shape characteristics well~\cite{pomerleau2013comparing, rusu2012downsampling, Rusu2011 }.
Recent schemes pro-actively chose samples to preserve geometric features like corners and edges~\cite{qi2019, chen2018, yang2015feature}, but do not ensure overall quality reconstruction of the original PC.
In this paper, we orchestrate a graph sampling method \cite{yuanchao2019ICASSP, bai2020} to choose 3D points in a PC to minimize a worst-case reconstruction error. 
\textit{To our knowledge, this is the first PC sub-sampling work that systematically minimizes a global reconstruction error metric.}

The \textit{graph sampling} problem in the field of \textit{graph signal processing} (GSP) \cite{ortega18ieee,cheung18} bears strong resemblance to the PC sub-sampling problem: 
select a subset of graph nodes to collect signal samples in order to minimize a reconstruction error, assuming that the target graph signal is smooth or bandlimited. 
The majority of existing graph sampling methods \cite{anis16, tsitsvero2016, chamon2018, chen2015_sampling} require computing extreme eigenvectors (corresponding to the smallest or largest eigenvalues) of a large adjacency or graph Laplacian (sub-)matrix per iteration, which is expensive.
In contrast, \cite{yuanchao2019ICASSP,bai2020} proposed an eigen-decomposition-free method called \textit{Gershgorin Disc Alignment Sampling} (GDAS)\footnote{GDAS is described along with an illustrative example in Section\;\ref{subsec:GDAS}.} based on the well-known \textit{Gershgorin Circle Theorem} (GCT) \cite{varga04}.
In summary, the goal is to maximize the lower bound of the smallest eigenvalue $\lambda_{\min}(\B)$ of a coefficient matrix $\B = \H^{\top} \H + \mu \L$ in a linear system, where $\L$ is a combinatorial graph Laplacian matrix.
GDAS selects nodes via sampling matrix $\H$ to maximize the smallest Gershgorin disc left-end of a similar-transformed matrix $\C = \S \B \S^{-1}$ (with the same eigenvalues as $\B$), which is the smallest eigenvalue lower bound $\lambda^-_{\min}(\C) \leq \lambda_{\min}(\C) =  \lambda_{\min}(\B)$ by GCT. 

However, GDAS requires matrix $\L$'s disc left-ends to be initially aligned at the same value before graph nodes are chosen.
To use GDAS for PC sub-sampling, in this paper we first articulate a sampling objective by assuming a post-processing procedure that \textit{super-resolves} (SR) a sub-sampled PC back to full resolution based on \textit{feature graph Laplacian regularization} (FGLR) \cite{dinesh2020}.
The SR procedure amounts to solving a system of linear equations with coefficient matrix $\B = \H^{\top} \H + \mu \cL$, 
where $\cL$ is a symmetric, positive semi-definite (PSD) matrix computed from a neighborhood graph connecting 3D points.
In general, $\cL$ does not have its disc left-ends aligned at any one value. 
We prove that minimizing a worst-case SR reconstruction error is equivalent to maximizing the smallest eigenvalue $\lambda_{\min}$ of $\B$.

Instead of maximizing $\lambda_{\min}$ of $\B$ directly, for fast computation we maximize a lower bound $\lambda^-_{\min}$ of $\B$ efficiently in three steps.
First, we interpret matrix $\cL$ as a generalized graph Laplacian matrix for an \textit{unbalanced} signed graph $\cG$ (with positive and negative edges). 
We propose an optimization algorithm to approximate $\cG$ with a \textit{balanced} graph $\cG_B$, with the corresponding generalized graph Laplacian matrix $\cL_B$, based on the well-known \textit{Cartwright-Harary Theorem} (CHT)~\cite{easley2010networks} for balanced graphs.

Second, we leverage on a recent theorem called \textit{Gershgorin disc perfect alignment} (GDPA)
\cite{yang2020} that proves disc left-ends of a generalized graph Laplacian matrix $\M$ for a balanced and irreducible signed graph can be \textit{perfectly} aligned at $\lambda_{\min}(\M)$ via a similarity transform $\S \M \S^{-1}$, where $\S = \text{diag}(1/v_1, 1/v_2, \ldots)$, and $\v=[v_1\hspace{4pt} v_2\hspace{4pt} \ldots]^{\top}$ is the first eigenvector of $\M$. 
Using this theorem, we compute similarity transform $\cL_p = \S \cL_B \S^{-1}$ so that disc left-ends of $\cL_p$ are aligned at $\lambda_{\min}(\cL_B)$, where $\S = \text{diag}(1/v_1, 1/v_2, \ldots)$ and $\v$ is the first eigenvector of $\cL_B$, computed using a known fast algorithm\footnote{Though our proposal computes the first eigenvector $\v$ of $\cL_B$ once per PC, it is still much faster than previous graph sampling algorithms \cite{anis16, tsitsvero2016, chamon2018, chen2015_sampling} that require computing extreme eigenvector(s) of a large matrix per sample.}  \textit{Locally Optimal Block Preconditioned Conjugate Gradient} (LOBPCG)~\cite{knyazev2001}. 
Finally, we adopt GDAS \cite{yuanchao2019ICASSP,bai2020} to choose 3D points for PC sub-sampling to maximize $\lambda^-_{\min}(\H^{\top} \H + \mu \cL_p)$ via selection of $\H$.
Extensive experimental results show that our PC sub-sampling method outperforms competing methods~\cite{Ranzuglia2012, corsini2012efficient, pomerleau2013comparing, rusu2012downsampling, shi2011, yang2015feature, qi2019, dinesh2020ICIP} in super-resolved quality both numerically evaluated with two common PC metrics~\cite{tian2017}, and visually for a wide range of PC datasets.

The outline of the paper is as follows.
We first discuss related works in PC sub-sampling in Section\;\ref{sec:related}.
We review basic concepts in GCT, GSP and signed graph balancing in Section\;\ref{sec:prelim}.
We describe our PC SR algorithm, a simplified variant of \cite{dinesh19}, in Section\;\ref{sec:SR}.
We present our PC sub-sampling strategy in Section\;\ref{sec:Sampling}.
We formulate our signed graph balancing problem and present our graph balancing algorithm in Section\;\ref{sec:balance} and \ref{sec:algo}, respectively.
Finally, we present results and conclusions in Section\;\ref{sec:exp} and \ref{sec:conclude}, respectively.

\section{Related Work}
\label{sec:related}
\label{sec:related}

We categorize existing methods for point cloud (PC) sub-sampling into two groups: indirect 
and direct. 
Given a sampling budget to choose 3D points, direct methods select points from a given PC directly, while indirect methods first construct a triangular or polygonal mesh from a given PC, and then perform node merging while preserving local structure. 
Examples of indirect methods include Poission disk sub-sampling (PDS)~\cite{corsini2012efficient}, mesh element sub-sampling (MESS), Monte-Carlo sub-sampling (MCS), and stratified triangle sub-sampling (STS) \cite{Ranzuglia2012}. 
There are two main limitations of indirect methods: 1) mesh construction is computationally expensive, and 2) mesh-based sub-sampling fundamentally changes the locations of original 3D points, which can cause geometric distortions~\cite{chen2018}. 

Based on the core techniques used, we divide direct methods into five categories\footnote{Some of these categories can also be found in~\cite{qi2019}.}:
random sub-sampling, averaging-based sub-sampling, clustering-based sub-sampling, progression-based sub-sampling, and, GSP-based sub-sampling.\\
\noindent \textbf{Random Sub-sampling:}
These methods~\cite{pomerleau2013comparing,anderson2005lidar} randomly choose a subset of points from a given PC using a uniform probability distribution, often resulting in distorted geometric features like edge, valleys, and corners.\\
\noindent\textbf{Averaging-based Sub-sampling:}
\cite{ryde20103d, loop2013real, rusu2012downsampling, Rusu2011} partition the 3D space in which a PC resides into 3D boxes (called \textit{voxels}), then quantize all the points in each voxel to its centroid. 
Similarly, \cite{peng2005geometry, hornung2013octomap} use an octree representation of PC to perform sub-sampling. 
Due to point averaging per voxel or octree branch, geometric features like edges and valleys are often coarsened and distorted.\\ 
\noindent\textbf{Clustering-based Sub-sampling:}
In these methods, first, a given PC is divided into non-overlapping clusters. 
Then, one or more points are chosen to represent each cluster. 
In~\cite{yu2010asm}, points are clustered  locally to preserve geometry features, while~\cite{shi2011} employs the $k$-means clustering algorithm to gather similar points together in the spatial domain, and then uses the maximum normal vector deviation to further subdivide a given cluster into smaller sub-clusters. 
In~\cite{benhabiles2013}, a coarse-to-fine approach is proposed to obtain a coarse cloud using volumetric clustering. However, these clustering-based sub-sampling causes distortion in intricate geometric details.\\  \noindent\textbf{Progression-based Sub-sampling:}
The farthest point resampling method is proposed in~\cite{moenning2004}, which selects the points progressively according to Voronoi diagrams. In~\cite{song2009}, the importance of a point is defined in terms of its neighboring points and their surface normals. Then the least important points are removed, and the importance of remaining points are updated. Further, in \cite{yang2015}, the mean curvature of points is computed based on \textit{Principal Component Analysis} (PCA) and Fourier Transform. Then the points around the point with the largest curvature are progressively removed. However, in some cases, geometric feature like edges can be distorted due to errors in finding important points.
\\              
\noindent\textbf{GSP-based Sub-sampling:}
In~\cite{chen2018}, a graph filter-based operator is proposed to extract application-dependent features of a given PC. 
Then samples from a given PC are found to minimize a feature reconstruction error. If the overall PC reconstruction error is the objective, then~\cite{chen2018} defaults to a uniform random sub-sampling scheme, which often leads to distorted geometric features like edges, valleys, and corners as previously described. On the other hand, if the objective is to minimize the reconstruction error of  geometric features, the approach in~\cite{chen2018} may result in extremely non-uniform sampling, resulting in distortion in the overall PC reconstruction shape. 
The approach in~\cite{qi2019} strikes a balance between preserving sharp features and keeping uniform density during sampling. However, the issue of how to best balance between feature preservation and near-uniform sub-sampling in order to minimize global reconstruction error is left unaddressed.

\noindent \textbf{Our Previous Work}: 
This work is an 
extension of our recent conference paper~\cite{dinesh2020ICIP}, which presented a PC sub-sampling strategy based on fast graph sampling GDAS \cite{yuanchao2019ICASSP}.
The key contribution in~\cite{dinesh2020ICIP} is an ad-hoc graph balancing procedure to construct a balanced graph $\cG_B$ given an arbitrary signed graph $\cG$---a pre-processing step prior to GDAS.
Instead of the ad-hoc procedure in \cite{dinesh2020ICIP}, in this paper we design an optimization algorithm to construct $\cG_B$ such that 
$\lambda_{\min}(\H^{\top} \H + \mu \cL_B)$ forms a tight lower bound for $\lambda_{\min}(\H^{\top} \H + \mu \cL)$ in order to minimize the worst-case PC-SR reconstruction error from chosen samples. 
We show in Section\;\ref{sec:exp} that our graph balance algorithm outperforms the ad-hoc procedure in \cite{dinesh2020ICIP} noticeably in SR reconstruction error.
\section{Preliminaries}
\label{sec:prelim}

We first review basic concepts in Gershgorin Circle Theorem (GCT), graph signal processing (GSP), and graph balance, which are needed to understand our proposed graph-based PC sub-sampling strategy.

\subsection{Gershgorin Circle Theorem (GCT)} 

Denote by $\F \in \mathbb{R}^{N\times N}$ a real square matrix. 
Corresponding to each row $i$ in $\F$, we define a \textit{Gershgorin disc} with radius $r_i~=~\sum_{j\neq i}|\F(i,j)|$ and center $c_i~=~\F(i,i)$. 
GCT~\cite{varga04} states that each eigenvalue $\lambda$ of $\F$ lies within at least one Gershgorin disc $i$, \textit{i.e.}, $\exists i$ such that $c_i - r_i \leq \lambda \leq c_i + r_i$. 
Thus, we can write GCT lower and upper bounds for minimum and maximum eigenvalues of $\F$ as
\begin{equation}
  \lambda_{\text{min}}(\F)\geq \min_{i}\{c_i - r_i\}, ~~~ \lambda_{\text{max}}(\F)\leq \max_{i}\{c_i + r_i\}.
\end{equation}

\subsection{Graph Signal Processing}

Consider an undirected graph $\mathcal{G}~=~(\mathcal{V}, \mathcal{E},\mathcal{U})$ composed of a node set $\mathcal{V}$ of size $N$ and an edge set $\mathcal{E}$ specified by $(i,j,w_{i,j})$, $i\neq j$, where $i,j\in\mathcal{V}$ and $w_{i,j}\in\mathbb{R}$ is the edge weight that encodes the (dis)similarity or (anti-)correlation between nodes $i$ and $j$, depending on the sign of $w_{i,j}$. 
In addition, each node $i$ may have a self-loop $(i)\in\mathcal{U}$ with weight $u_{i}\in\mathbb{R}$.

One can characterize $\cG$ with a symmetric \textit{adjacency matrix} $\mathbf{W} \in \mathbb{R}^{n \times n}$, where $\W(i,j)~=~\W(j,i)~=~w_{i,j}$, $(i,j)\in\mathcal{E}$ and $\W(i,i)~=~u_i$, $(i)\in\mathcal{U}$. 
A diagonal \textit{degree matrix} $\mathbf{D}$ has diagonal entries $\D(i,i) = \sum_j \W(i,j)$, and off-diagonal entries $\D(i,j) = 0, \forall i \neq j$. 
Given $\mathbf{W}$ and $\mathbf{D}$, the \textit{combinatorial graph Laplacian matrix}~\cite{chung1997} is defined as~$\mathbf{L}~=~\mathbf{D}-\mathbf{W}$.
Further, a \textit{generalized graph Laplacian matrix}~\cite{Biyikoglu2005} that accounts for self-loops in $\cG$ is defined as $\M~=~\L+\text{diag}(\W)$, where $\text{diag}(\W)$ is a diagonal matrix collecting the diagonal entries of $\W$ representing self-loops in $\cG$. 
Note that any real symmetric matrix can be interpreted as a generalized graph Laplacian matrix $\M$ of some undirected graph.

A vector $\x \in \mathbb{R}^N$ can be interpreted as a \textit{graph signal}, where $x_i$ is a scalar value assigned to node $i\in\mathcal{V}$ in a given graph $\cG$. 
A graph variation term $\x^{\top}\L\x$ can then be written as \cite{chung1997}
\begin{equation}
 \x^{\top}\L\x=\sum_{i,j\in\mathcal{E}}w_{i,j}(x_i-x_j)^{2}.
 \label{eq:GLR}
\end{equation}
If $\cG$ contains only positive edges, then $\L$ is \textit{positive semi-definite} (PSD) \cite{cheung18}, and smoothness of $\x$ with respect to (w.r.t.) $\cG$ can be measured using $\x^{\top}\L\x \geq 0$, called the \textit{graph Laplacian regularizer} (GLR) in the GSP literature~\cite{shuman2013}. 
Extending GLR, we recently proposed a new regularization term called \textit{feature graph Laplacian regularizer} (FGLR) \cite{dinesh2020}---$\mathbf{h}(\mathbf{x})^{\top} \mathbf{L} \, \mathbf{h}(\mathbf{x})$---where $\mathbf{h}(\mathbf{x}) \in \mathbb{R}^N$ is a vector function of $\mathbf{x}$ that computes a \textit{feature vector}, which is assumed to be smooth w.r.t. the graph characterized by Laplacian $\L$. 
FGLR can be used for PC restoration, where $\mathbf{h}(\mathbf{x})$ computes surface normals at the 3D points with coordinates specified by $\mathbf{x}$ (see~\cite{dinesh2020} for details).

\subsection{Balance of Signed Graph}

A \textit{signed graph} is a graph with both positive and negative edge weights.
The concept of balance in a signed graph was used in many scientific disciplines, such as psychology, social networks and data mining~\cite{leskovec2010}. 
In this paper, we adopt the following definition of a \textit{balanced signed graph}~\cite{easley2010networks, yang2020}:

\begin{definition}
A signed graph $\cG$ is balanced if $\cG$ does not contain any cycle with odd number of negative edges. 
\end{definition}

For intuition, consider a 3-node graph $\cG$, shown in Fig.\;\ref{fig:3nodes}. 
Given the edge weight assignment in  Fig.\;\ref{fig:3nodes}(left), we see that this graph is balanced; the only cycle has an even number of negative edges (two).
Note that nodes can be grouped into \textit{two} clusters---red cluster $\{i,k\}$ and blue cluster $\{j\}$---where same-color node pairs are connected by positive edges, and different-color node pairs are connected by negative edges.

In contrast, the edge weight assignment in Fig.\;\ref{fig:3nodes}(right) produces a cycle of odd number of negative edges (one). 
This graph is not balanced.
In this case, nodes cannot be assigned to two colors such that positive edges connect same-color node pairs, and negative edges connect different-color node pairs. 
Generalizing from this example, we state the well-known \textit{Cartwright-Harary Theorem} (CHT)~\cite{easley2010networks} as follows.

\begin{theorem}
\label{theorem:CHT}
A given signed graph is balanced if and only if its nodes can be colored into red and blue such that a positive edge always connects nodes of the same color, and a negative edge always connects nodes of opposite colors. 
\end{theorem}

Thus, to determine if a given graph $\cG$ is balanced, instead of examining all cycles in $\cG$ and checking if each contains an odd or even number of negative edges (exponential time in the worst case), we can check if nodes can be colored into blue and red with consistent edge signs, as stated in Theorem~\ref{theorem:CHT}, which can be done in linear time~\cite{yang2020}. 
In this paper, we use Theorem~\ref{theorem:CHT} to design an algorithm to augment edges in a graph $\cG$ until the resulting graph becomes balanced.

\begin{figure}[t]
\centering
\includegraphics[width=0.43\textwidth]{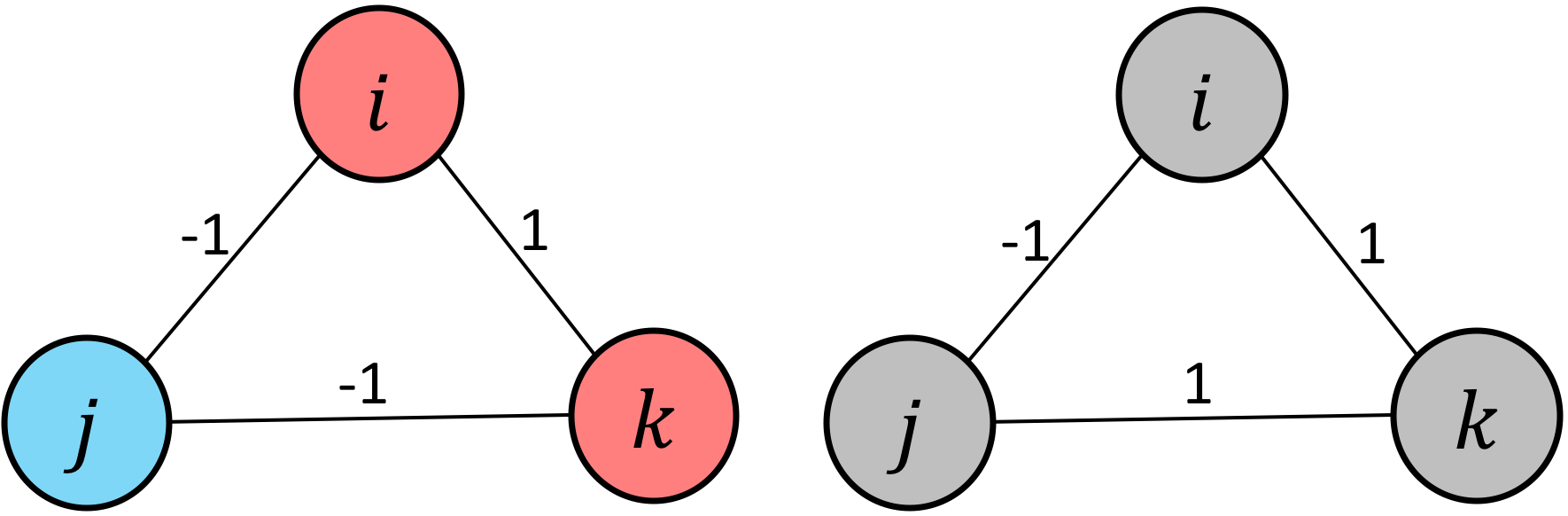}
\vspace{-6pt}
\caption{Examples of three-nodes signed graphs: balanced (left) and unbalanced (right)}
\label{fig:3nodes}
\vspace{-10pt}
\end{figure}

\section{Point Cloud Super-Resolution}
\label{sec:SR}
To mathematically define an objective for PC sub-sampling, we first assume that the sub-sampled PC will be subsequently \textit{super-resolved} from $m$ 3D point samples back to $n$ points, where $m < n$, via a simple PC super-resolution (SR) method (a simplified variant of \cite{dinesh19}). We chose this SR method because the resulting SR problem is an unconstrained quadratic programming (QP) optimization \eqref{eq:obj1}, with a system of linear equations as solution. 
Given this solution, we then select a sampling matrix~\cite{wang18} $\H$ to minimize a worst-case SR reconstruction error. 
Note that though we use a specific SR method to define an objective for PC sub-sampling, we will demonstrate in Section\;\ref{sec:exp} that our sub-sampling algorithm performs significantly better than competitors when combined with alternative SR methods~\cite{guennebaud2008, huang2013, dinesh2019, dinesh2020ICASSP} as well.

Since the SR method described here is graph-based, we first construct a graph given a PC as input.
We then define a fidelity term and a graph signal prior, so that a QP optimization problem can be formulated.

\subsection{Graph Construction for 3D Point Cloud}
\label{sec:graph_cons}

To enable graph-based processing of $n$ 3D points, we first construct an $n$-node graph.
For simplicity, we assume here that the original $n$ 3D points are used to construct an $n$-node graph. 
In practice, to construct an $n$-node graph from $m$ 3D point samples, $n-m$ interior 3D points can be interpolated from $m$ samples and then subsequently refined, as done in~\cite{dinesh19,dinesh2020ICASSP}. In particular, we construct an undirected positive graph $\mathcal{G}_p=(\mathcal{V}_{p}, \mathcal{E}_{p})$ composed of a node set $\mathcal{V}_{p}$ of size $n$ (each node represents a 3D point) and an edge set $\mathcal{E}_{p}$ specified by $(i,j,w_{i,j})$, where $i \neq j$,  $i,j\in\mathcal{V}$ and $w_{i,j}\in\mathbb{R}^{+}$ with no self-loops. 
We connect each 3D point (graph node) to its $k$ nearest neighbors $j$ in Euclidean distance, so that each point can be filtered with its $k$ neighboring points under a graph-structured data kernel \cite{ortega18ieee,cheung18}. 

Edge weight $w_{ij}$ between nodes $i$ and $j$ is computed based on pairwise Euclidean distance and surface normal difference, \ie,
\begin{equation}
w_{ij}=
 \exp \left\{-\frac{\norm{\mathbf{p}_{i}-\mathbf{p}_{j}}_{2}^{2}}{\sigma_{p}^{2}}-\frac{\norm{\mathbf{n}_{i}-\mathbf{n}_{j}}_{2}^{2}}{\sigma_{n}^{2}}\right\} 
\label{eq:edge_weight}
\vspace{-5pt}
\end{equation}
where $\mathbf{p}_{i}\in \mathbb{R}^{3}$ and $\mathbf{n}_{i}\in \mathbb{R}^{3}$ are the 3D coordinate and the unit-norm surface normal of point $i$, respectively, and  $\sigma_{p}$ and $\sigma_{n}$ are parameters.
In words, \eqref{eq:edge_weight} states that edge weight $w_{ij}$ is large (close to 1) if 3D points $i, j$ are close in distance and associated normal vectors are similar, and small (close to 0) otherwise. 
The weight definition \eqref{eq:edge_weight} is similar to \textit{bilateral filter} weights \cite{tomasi98} with domain and range filters. 
We have demonstrated in our previous work~\cite{dinesh2020} experimentally and mathematically that weight definition \eqref{eq:edge_weight} leads to 
FGLR that promotes piecewise smooth geometry reconstruction in 3D PC.

\subsection{Fidelity Term}
Denote by $\mathbf{p}~=~\begin{bmatrix} \mathbf{p}_{1}^{\top} & \hdots & \mathbf{p}_{{n}}^{\top} \end{bmatrix}^\top\in\mathbb{R}^{3n}$ the position vector of the super-resolved $n$ 3D points in the PC, where $\mathbf{p}_{i}\in\mathbb{R}^{3}$ is the \mbox{$x$-,} $y$-, and $z$  coordinate of point $i$.
Denote by $\H \in \{0, 1\}^{3m \times 3n}$ the sampling matrix~\cite{wang18} that selects $m$ samples from $n$ available 3D points $\p$.
Because each time a 3D point $i$ is selected as sample, its $x$-, $y$- and $z$-coordinates must be extracted together. Hence, sampling matrix $\H$ can be written as $\H=\tilde{\H}\otimes\mathbf{I}_{3}$, where $\tilde{\H} \in \{0,1\}^{m \times n}$ is
\begin{equation}
   \tilde{\H}(i,j)=\begin{cases}
        1, \text{if} \hspace{4pt} i\text{-th selected point is}\hspace{4pt} \p_{j}\\
        0, \text{otherwise}.
    \end{cases}
\end{equation}
Here, $\mathbf{I}_3$ is the $3\times 3$ identity matrix and $\otimes$ is the Kronecker product.

We assume that the $m$ sub-sampled points $\q$ are not perfectly accurate, \eg, they were corrupted by zero-mean, independent and identically distributed (iid) Gaussian noise of variance $\sigma^2$, when SR is performed for PC recovery\footnote{In practice, PCs are constructed from inaccurate depth measurements in active sensors or stereo-matching using color image pairs, both of which result in noisy 3D points. See \cite{dinesh2020} for a detailed discussion.}. 
Thus, the formation model for the observed sub-sampled points $\q$ is defined as
\begin{equation}
 \q=\H\p+\e,    
\end{equation}
where $\e\in\mathbb{R}^{3m}$ is the additive noise.
Thus, given $\H$, to ensure that the reconstructed PC $\p$ is consistent with noise-corrupted samples $\q$, we write the fidelity term as
$\|\q - \H \p \|_2^2$.

\subsection{Signal Prior}
\label{sec:FGLR}

Denote by $\n = [\n_x^{\top} \; \n_y^{\top} \; \n_z^{\top} ]^{\top} \in \mathbb{R}^{3n}$ a vector of $x$\nobreakdash-, $y$-, and $z$-coordinates of the surface normals of the super-resolved PC. 
Note that $\mathbf{n}_{x}$, $\mathbf{n}_{y}$, and $\mathbf{n}_{z}$ are vector functions (or \textit{features}) of 3D point coordinates $\mathbf{p}$.
Hence, we can rewrite $\n(\p)~=~[\mathbf{n}_{x}^{\top}(\mathbf{p}) \hspace{4pt} \mathbf{n}_{y}^{\top}(\mathbf{p}) \hspace{4pt} \mathbf{n}_{z}^{\top}(\mathbf{p})]^{\top}$, as done in \cite{dinesh2020}.

As done in \cite{dinesh2018fast, dinesh2019, dinesh2020, dinesh2020ICASSP, dinesh2020ICIP}, we define \textit{smoothness in 3D geometry} in terms of surface normals at 3D points. 
Specifically, we assume that $x$-, $y$-, and $z$-coordinates of surface normals, $\n_{x}(\p)$, $\n_{y}(\p)$ and $\n_{z}(\p)$, are smooth w.r.t. the graph $\cG_{p}$ constructed in Section~\ref{sec:graph_cons}.  
Mathematically, we write a signal prior $\text{Pr}(\p)$ based on FGLR\footnote{There are other possible signal priors in the literature for PC processing.
See \cite{dinesh2020} for a comparison of FGLR against alternative priors.}~\cite{dinesh2020} to promote smoothness of $\mathbf{n}_x(\mathbf{p})$, $\mathbf{n}_y(\mathbf{p})$ and $\mathbf{n}_z(\mathbf{p})$ as follows:
\vspace{0pt}
\begin{equation}
\begin{aligned}
\text{Pr}(\mathbf{p})&=\mathbf{n}_x^{\top}(\mathbf{p})\mathbf{\tilde{L}}_{p}\mathbf{n}_{x}(\mathbf{p})+\mathbf{n}_y^{\top}(\mathbf{p})\mathbf{\tilde{L}}_{p}\mathbf{n}_y(\mathbf{p})+\mathbf{n}_z^{\top}(\mathbf{p})\mathbf{\tilde{L}}_{p}\mathbf{n}_{z}(\mathbf{p}),\\
&=\mathbf{\n}^{\top}({\mathbf{p}})\mathbf{L}_{p}\mathbf{\n}({\mathbf{p}}) 
\label{eq:main_FGLR}
\end{aligned}
\end{equation}
where $\mathbf{\tilde{L}}_{p}$ is the PSD combinatorial graph Laplacian matrix for the constructed positive graph $\cG_p$
and $\mathbf{L}_{p}~=~\mathrm{diag}(\mathbf{\tilde{L}}_{p}, \mathbf{\tilde{L}}_{p}, \mathbf{\tilde{L}}_{p})$.
Note that $\L_{p}$ is PSD, given $\mathbf{\tilde{L}}_{p}$ is PSD.

\subsection{Optimization Formulation}

We formulate an optimization problem by combining the fidelity term and the signal prior~(\ref{eq:main_FGLR}) together:
\begin{align}
\min_{\p} \|\q - \H \p \|_2^2 + 
\mu \mathrm{Pr}(\p)
\label{eq:obj0}
\end{align}
where $\mu > 0$ is a parameter trading off the importance between the fidelity term and the signal prior.

Using state-of-the-art surface normal estimation methods~\cite{huang2001, klasing2009}, $\mathbf{n}_{x}(\mathbf{p})$, $\mathbf{n}_{y}(\mathbf{p})$, and $\mathbf{n}_{z}(\mathbf{p})$ are non-linear vector functions of the optimization variable $\mathbf{p}$ in~(\ref{eq:obj0}).
This non-linearity makes the optimization problem in (\ref{eq:obj0}) difficult to solve. 
However, using a linear approximation in \cite{dinesh2018fast,dinesh2020}, we can locally approximate $\n_{x}(\mathbf{p})$, $\n_{y}(\mathbf{p})$, and $\n_{z}(\mathbf{p})$ as linear functions of $\p$:
\begin{equation}
\small
\n_{x}(\mathbf{p}) \approx \A_{x} \p + \b_{x},\hspace{5pt} \n_{x}(\mathbf{p}) \approx \A_{x} \p + \b_{x},\hspace{5pt} \n_{x}(\mathbf{p}) \approx \A_{x} \p + \b_{x},
\label{eq:normalApprox}
\end{equation}
where $\A_{x}, \A_{y}, \A_{z} \in \mathbb{R}^{n \times 3n}$ and $\b_{x}, \b_{y}, \b_{z} \in \mathbb{R}^{n}$ are computed as done in \cite{dinesh2020}.
Note that by construction, matrices $\A_{x}$, $\A_{y}$, and $\A_{z}$ are sparse. 
Using \eqref{eq:normalApprox},  FGLR in \eqref{eq:main_FGLR} can be rewritten as
\begin{equation}
\begin{split}
 \text{Pr}(\p)&=(\A \p + \b)^{\top} \L_{p} (\A \p + \b)\\
 &=\p^{\top} \underbrace{\A^{\top} \L_{p} \A}_{\cL} \p + 2 \underbrace{\b^{\top} \L_{p} \A}_{\c^{\top}} \p + \b^{\top} \L_{p} \b
\end{split}
\label{eq:obj1}
\end{equation}
where $\A = [\A_{x}^{\top} \hspace{5pt} \A_y^{\top} \hspace{5pt} \A_z^{\top}]^{\top}$, and $\b = [\b_x^{\top} \hspace{5pt} \b_y^{\top} \hspace{5pt} \b_z^{\top}]^{\top}$.
Note that matrix $\cL$ is PSD, since $\L_{p}$ is PSD.

Since all terms in \eqref{eq:obj1} are quadratic, the optimization problem in \eqref{eq:obj0} is an unconstrained QP problem, and the solution $\p^*$ can be found by taking the derivative w.r.t. $\p$ and setting it to 0, resulting in
\begin{align}
\underbrace{\left( \H^{\top} \H + \mu \cL \right)}_{\B} \p^* =
\H^{\top} \q + \mu \c.
\label{eq:sysmLin}
\end{align}
Equation~\eqref{eq:sysmLin} is a system of linear equations with a well-defined unique solution if $\B$ is \textit{positive definite} (PD). 
Since both $\H^{\top} \H$ and $\cL$ are PSD, by Weyl's Theorem~\cite{matrixAnalysis} $\B$ is also PSD, hence $\lambda_{\min}(\B) \geq 0$.
Since the optimization objective to be discussed in Section\;\ref{sec:sampling_obj} is to maximize $\lambda_{\min}(\B)$ via the choice of sampling matrix $\H$, for a sufficient number of samples in $\H$, $\lambda_{\min}(\B) > 0$ and $\B$ is PD, and thus \eqref{eq:sysmLin} has a unique solution. 

The stability of the linear system \eqref{eq:sysmLin}, however, depends on the \textit{condition number}: 
the ratio of the largest to smallest eigenvalues of the coefficient matrix $\B$, \ie, $\rho = \lambda_{\max}(\B) / \lambda_{\min}(\B)$. A smaller condition number translates to better numerical stability~\cite{el2002}. Thus,
we begin our discussion on how we select sampling matrix $\H$ to minimize $\rho$ next.
\vspace{-6pt}
\section{Point Cloud Sub-Sampling}
\label{sec:Sampling}
\subsection{Objective Definition}
\label{sec:sampling_obj}

Our sub-sampling objective is to minimize the condition number $\rho~=~\lambda_{\max}(\B) / \lambda_{\min}(\B)$ of the coefficient matrix $\B$ in \eqref{eq:sysmLin} via selection of the sampling matrix $\H$. 
Because $\H^{\top} \H$ is a diagonal matrix with 1's and 0's along its diagonal, $\lambda_{\max}(\H^{\top} \H)~=~1$.
Thus, by Weyl's Theorem \cite{matrixAnalysis}, $\lambda_{\max}(\H^{\top} \H + \mu \cL) \leq \lambda_{\max}(\H^{\top} \H) + \mu \lambda_{\max}(\cL)~=~1 + \mu \lambda_{\max}(\cL)$. 
Moreover, \cite{dinesh2020} showed that $\lambda_{\max}(\cL)$ can be upper-bounded by a sufficiently small value. 
Hence, $\lambda_{\max}(\B)$ is also upper-bounded by a small value unaffected by $\H$.

On the other hand, $\lambda_{\min}(\B)$ can be dangerously close to $0$ if $\H$ is not carefully chosen, resulting in a large $\rho$. 
In order to minimize $\rho$, we thus focus on maximizing $\lambda_{\min}(\B)$, \ie,
\begin{align}
\max_{\H \,|\, \mathrm{Tr}(\H^{\top}\H) \leq 3m} \lambda_{\min}(\B)
\label{eq:obj}
\vspace{-5pt}
\end{align}
where $\mathrm{Tr}(\F)$ is the trace of matrix $\F$. 
Similar to \cite{bai2020}, we can show that maximizing $\lambda_{\min}(\B)$ in  \eqref{eq:obj} is equivalent to minimizing a worst-case bound of \textit{mean square error} (MSE) between reconstructed point coordinates and original point coordinates.

\begin{proposition}
\label{pro:MMSE}
Assuming that $\cL$ and $\c$ in \eqref{eq:sysmLin} are  known during SR reconstruction\footnote{In practice, graph $\cG_{p}$ can be constructed as mentioned in Section~\ref{sec:graph_cons}. Hence $\cL$ and $\c$ can be computed approximately.},
maximizing $\lambda_{\min}(\B)$ in \eqref{eq:obj} minimizes a worst-case bound of MSE between reconstructed point coordinates $\mathbf{p}^{*}$ and original point coordinates $\p$.  
\end{proposition}
The proof of Proposition~\ref{pro:MMSE} is reported in Appendix A in the supplementary material.

\subsection{GDAS for Positive Graphs}
\label{subsec:GDAS}

We devise a graph sampling strategy to maximize objective \eqref{eq:obj}. 
To impart intuition, consider first the following simple case, from which we will generalize later. 
Assume first that $\cL$ is a combinatorial graph Laplacian matrix $\cL~=~\D - \W$, where $\D$ and $\W$ are the degree and adjacency matrices, respectively, for a weighted \textit{positive} graph $\cG$ without self-loops. 
In this case, all Gershgorin disc left-ends of $\cL$ are aligned at $0$, \ie, $c_i - r_i~=~\D(i,i) - \sum_{j} \W(i,j)~=~0~=~\lambda_{\min}(\cL), \forall i$. 
Thus, GDAS~\cite{yuanchao2019ICASSP,bai2020} can be employed to approximately solve \eqref{eq:obj} with roughly linear-time complexity. 

In a nutshell, GDAS maximizes the \textit{smallest eigenvalue lower bound} $\lambda_{\min}^-(\C)$ of a similar-transformed matrix $\C~=~\S \B \S^{-1}$ of $\B$, \ie, 
\begin{align}
\lambda_{\min}^-(\C) \leq \lambda_{\min}(\C) = \lambda_{\min}(\B),
\end{align}
since $\B$ and similar-transformed $\C$ have the same eigenvalues~\cite{varga2010gervsgorin}.
$\S = \text{diag}(s_1, \ldots, s_{3n})$ is a diagonal matrix with scalars $s_i \neq 0$ on its diagonal (thus $\S^{-1}$ is well defined).
By GCT \cite{varga04}, $\lambda_{\min}^-(\C) = \min_i (c_i - r_i)$ is the smallest Gershgorin disc left-end of matrix $\C$. 
Given a target smallest eigenvalue lower bound $T>0$, GDAS (roughly) realigns \textit{all} disc left-ends of $\C$ from $0$ to $T$ via two operations that are performed in a \textit{breadth-first search} (BFS) manner, expending as few samples $K$ as possible in the process.
Because $T$ and $K$ are proportional, if $K$ is larger than sample budget $m$, then $T$ must be decreased until $K \leq m$.
The largest target $T$ such that $K \leq m$ can be located efficiently via binary search.

The two basic disc operations for a given target $T$ are:
\begin{enumerate}
\item \textit{Disc Shifting}: select node $i$ in graph $\cG$ as sample, meaning that the $i$-th diagonal entry of $\H^{\top} \H$ becomes 1 and, as a result, shifts disc center $c_i$ of row $i$ of $\C$ to the right by $1$. 

\item \textit{Disc Scaling}: select scalar $s_i > 1$ in $\S$ to increase disc radius $r_i$ of row $i$ corresponding to sampled node $i$, thus decreasing disc radii $r_j$ of neighboring nodes $j$ due to $\S^{-1}$, and moving disc left-ends of nodes $j$ to the right. 
\end{enumerate}

We demonstrate these two disc operations in GDAS using an example from \cite{yuanchao2019ICASSP} as follows. 
Consider a 4-node line graph, shown in Fig.\;\ref{fig:linegraph}. Suppose we first sample node 3.  
If $\mu=1$, then the coefficient matrix $\B$ is as shown in Fig.\;\ref{fig:GDA_sampling}a, after $(3,3)$ entry is updated. Correspondingly, Gershgorin disc $\psi_3$'s left-end of row 3 moves from $0$ to $1$, as shown in Fig.\;\ref{fig:GDA_sampling}d (red dots and blue arrows represent disc centers and radii, respectively).

Next, we scale disc $\psi_3$ of sampled node $3$:~ apply scalar $s_{3} > 1$ to row 3 of $\B$ to move disc $\psi_3$'s left-end left to bound $T$ exactly, as shown in Fig.\;\ref{fig:GDA_sampling}b. 
Concurrently, scalar $1/s_{3}$ is applied to the third column, and thus the disc radii of $\psi_2$ and $\psi_4$ are reduced due to scaling of $w_{2,3}$ and $w_{4,3}$ by $1/s_{3}$. 
Note that diagonal entry $(3,3)$ of $\B$ (\ie,  $\psi_{3}$'s disc center) is unchanged, since the effect of $s_{3}$ is offset by $1/s_{3}$. 
We now see that by expanding disc $\psi_3$'s radius, the disc left-ends of its neighbors (\ie, $\psi_2$ and $\psi_4$) move right beyond bound $T$ as shown in Fig.\;\ref{fig:GDA_sampling}e.

Subsequently, scalar $s_2 > 1$ for disc $\psi_{2}$ can be applied to expand its radius, so that $\psi_2$'s left-end moves to threshold $T$.
Again, this has the effect of shrinking the disc radius of neighboring node $1$, pulling its left-end right beyond $T$, as shown in Fig.\;\ref{fig:GDA_sampling}c and Fig.~\ref{fig:GDA_sampling}f.
At this point, we see that left-ends of all Gershgorin discs have moved beyond threshold $T$, and thus we conclude that by sampling a single node $2$, smallest eigenvalue  $\lambda_{\min}$ of $\B$ is lower-bounded by $T$. 
See~\cite{yuanchao2019ICASSP, bai2020} for further details.

\begin{figure}[t]
\centering
\includegraphics[width=0.38\textwidth]{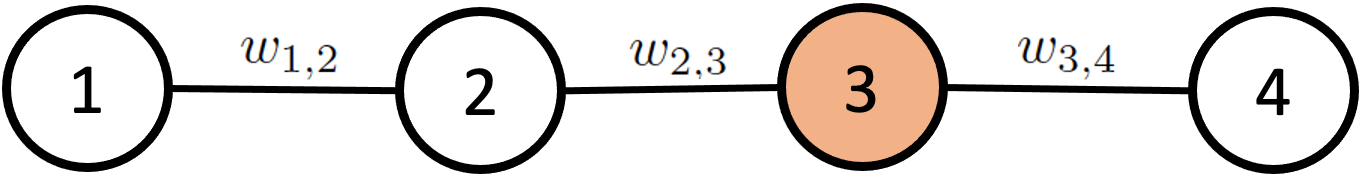}
\vspace{-7pt}
\caption{An illustrative example of a 4-node line graph from~\cite{yuanchao2019ICASSP}.}
\label{fig:linegraph}
\vspace{-12pt}
\end{figure}

\begin{figure}[t]
\centering
\subfloat[]{
\includegraphics[width=0.159\textwidth]{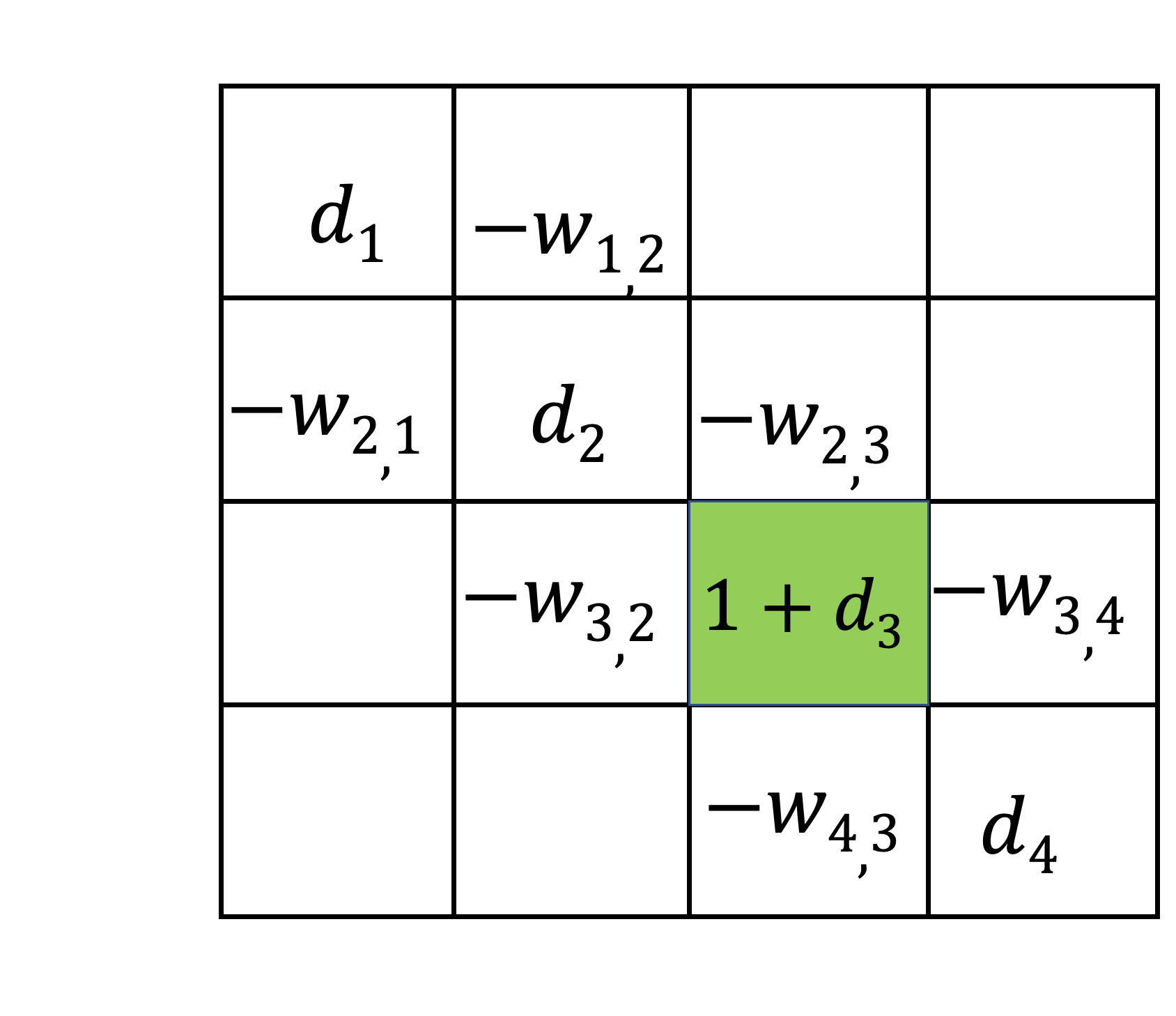}
\label{fig:pcd-on}}
\hspace{-8pt}
\subfloat[{}]{
\includegraphics[width=0.159\textwidth]{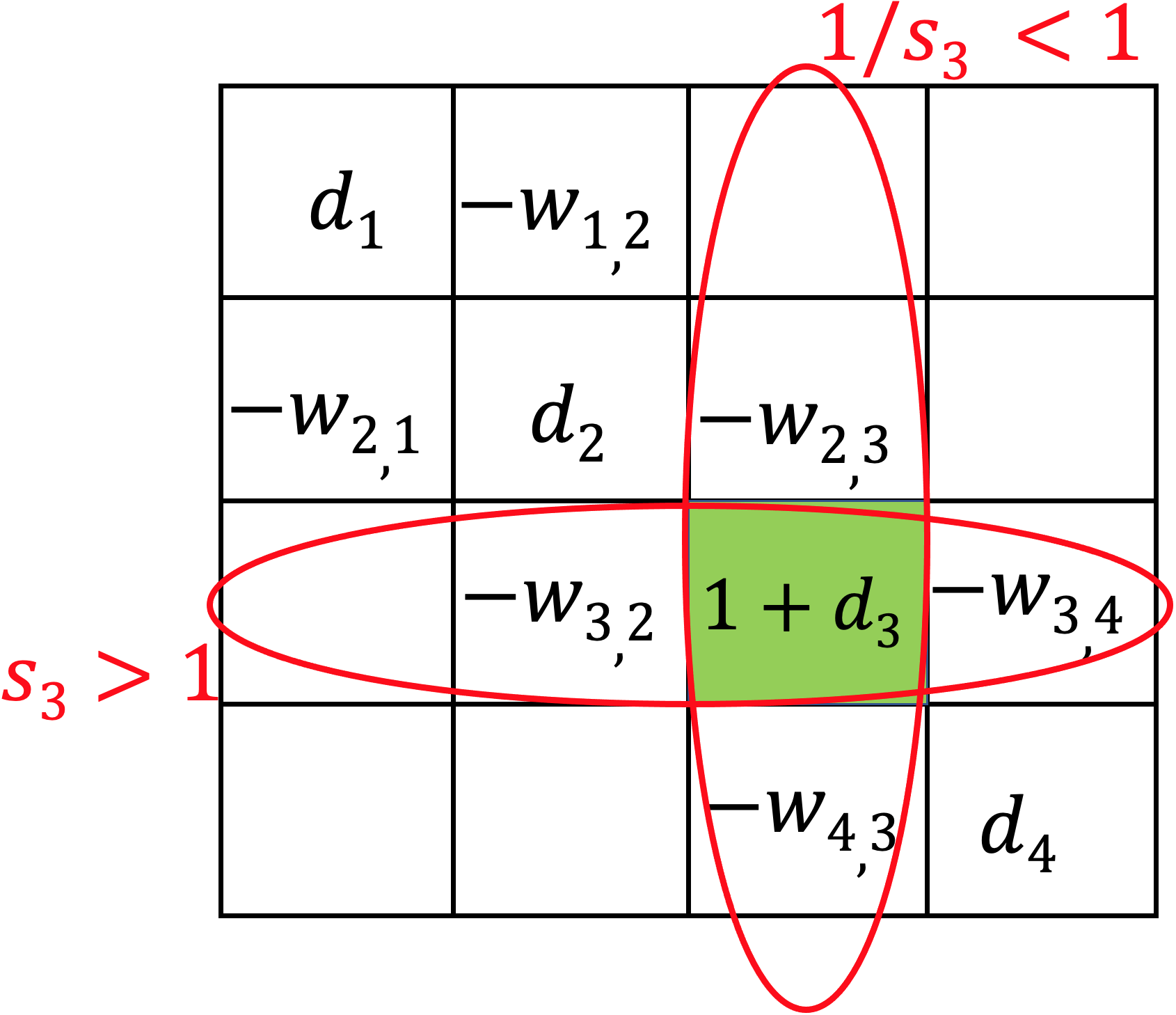}
\label{fig:pcd-off}}
\hspace{-8pt}
\subfloat[{}]{
\includegraphics[width=0.159\textwidth]{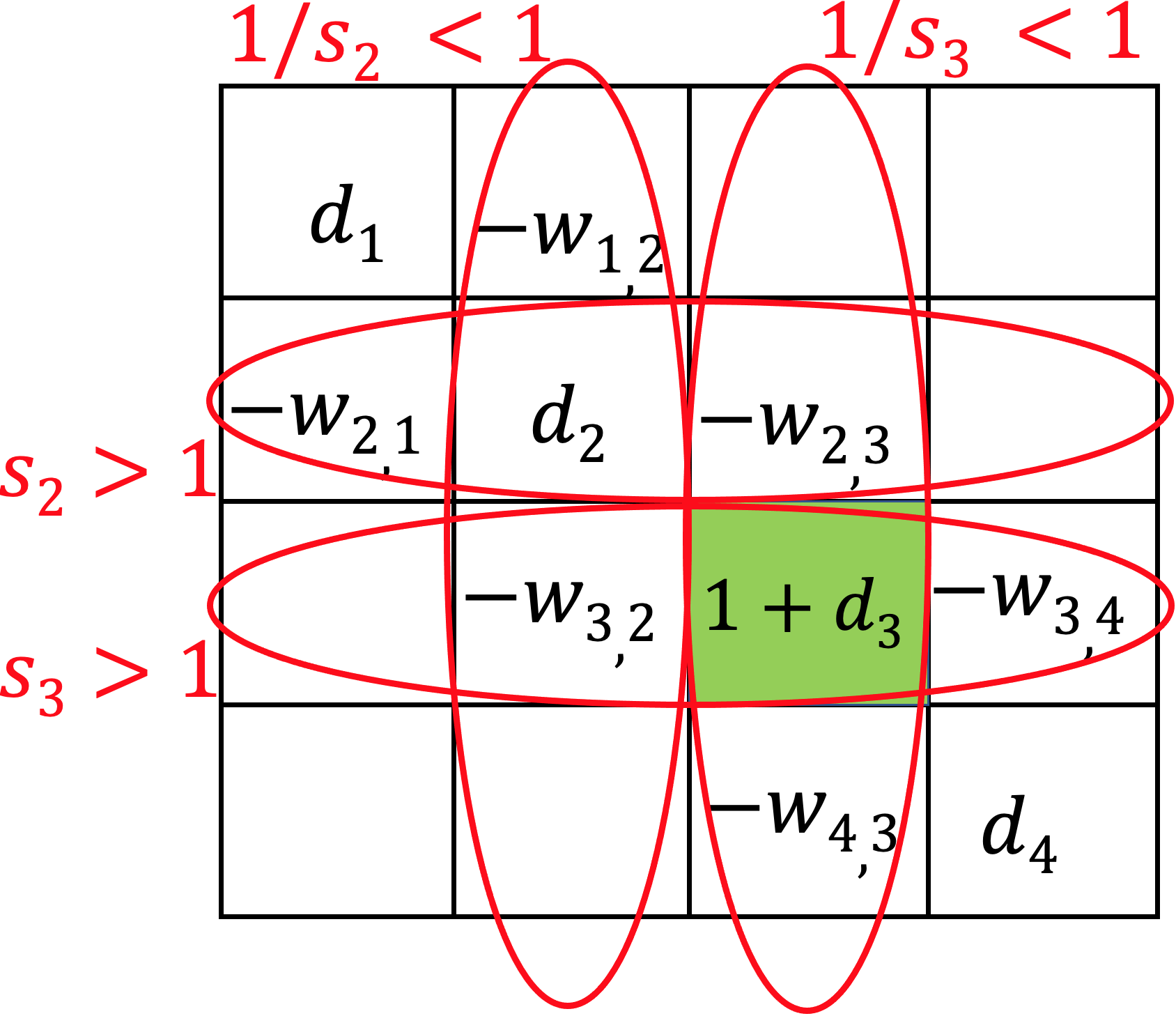}
\label{fig:pcd-off}}
\hspace{-8pt}\\
\vspace{-10pt}
\subfloat[]{
\includegraphics[width=0.149\textwidth]{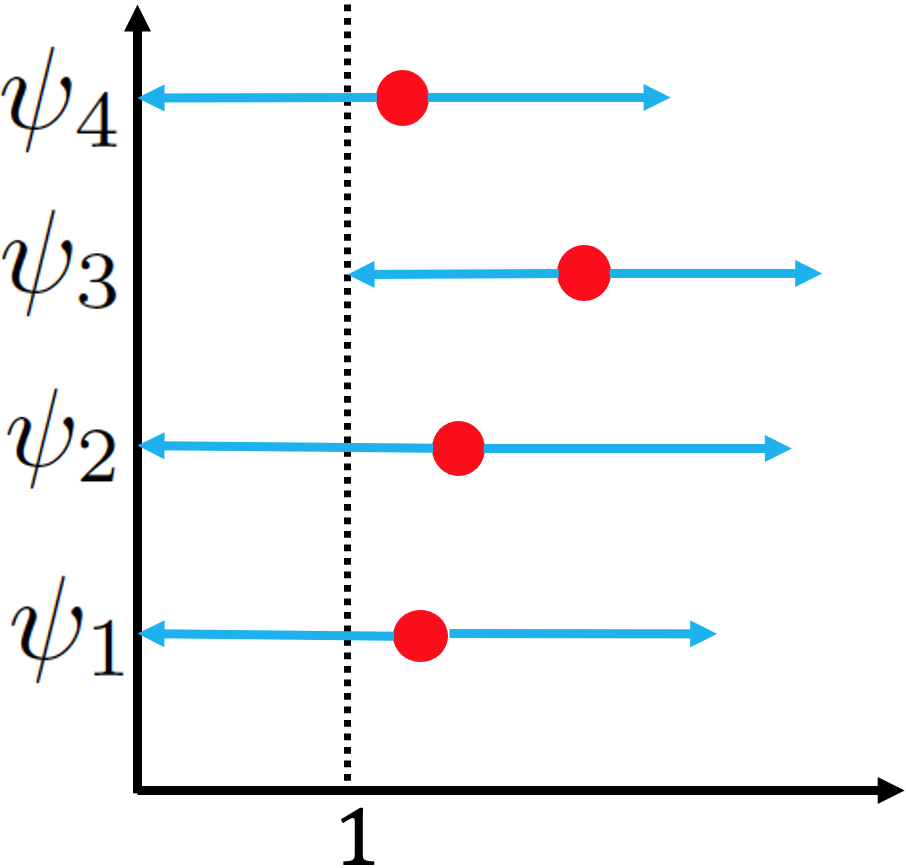}
\label{fig:pcd-on}}
\hspace{-8pt}
\subfloat[{}]{
\includegraphics[width=0.149\textwidth]{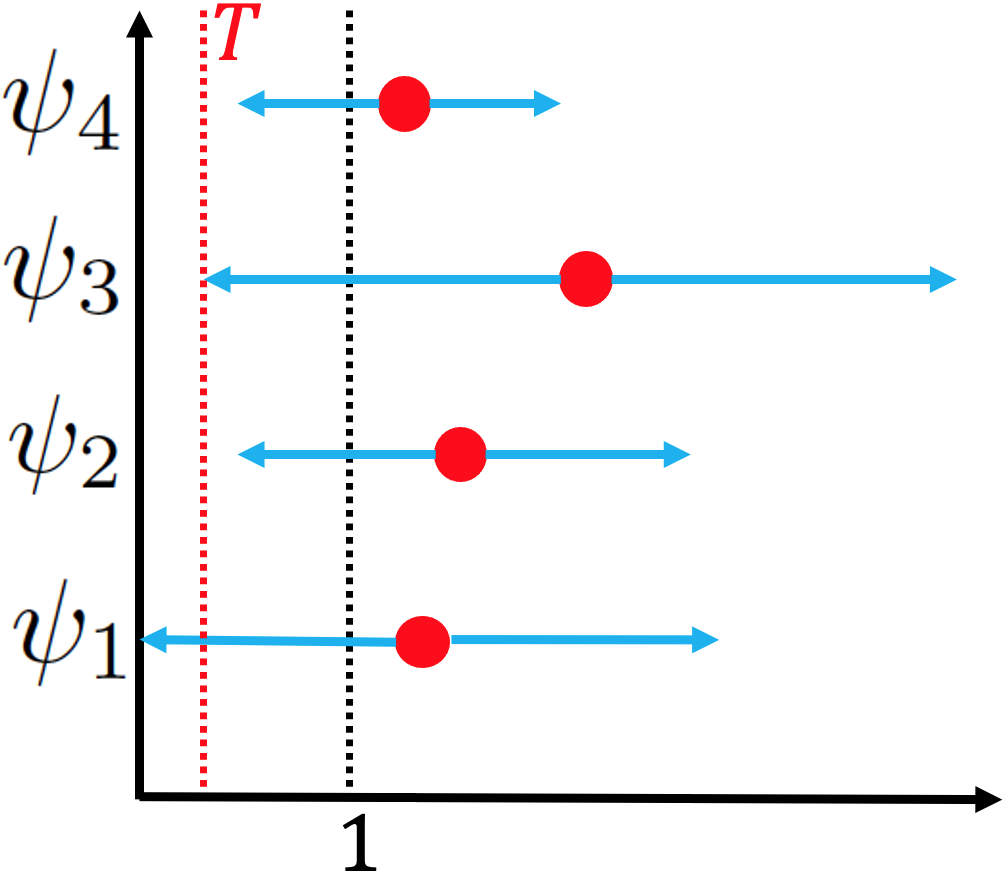}
\label{fig:pcd-off}}
\hspace{-8pt}
\subfloat[{}]{
\includegraphics[width=0.149\textwidth]{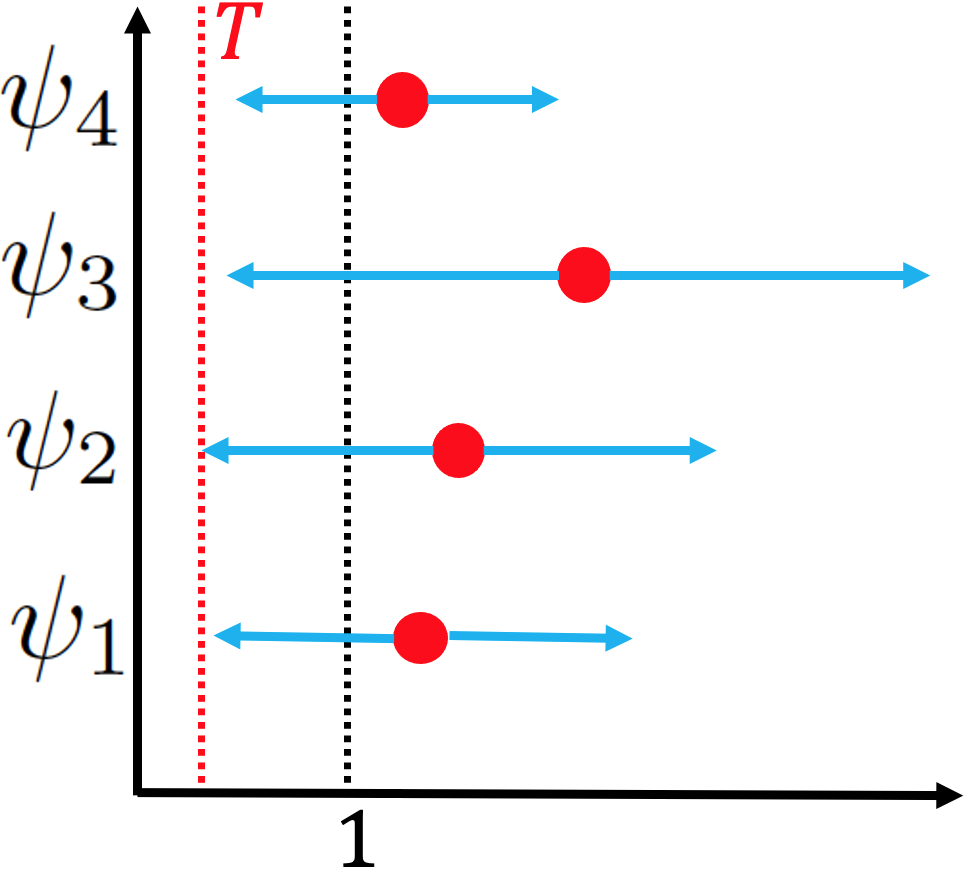}
\label{fig:pcd-off}}
\caption{An illustration of GDAS from~\cite{yuanchao2019ICASSP}. (a) sampling node 3. (b) scaling node 3. (c) scaling nodes 2 and 3. (d) Discs after sampling node 3. (e) Discs after scaling node 3. (f) Discs after scaling nodes 2 and 3.} 
\label{fig:GDA_sampling}
\vspace{-14pt}
\end{figure}

\subsection{GDAS for Signed Graphs}
\label{sec:GDA_signgraph}

In the previous example of GDAS, disc left-ends are initially aligned at the same value $0$, then via a sequence of disc shifting and scaling operations, become (roughly) aligned at target $T$.
In the general case, $\cL$ in \eqref{eq:obj} is not a combinatorial graph Laplacian for a weighted positive graph without self-loops, but is any generic symmetric real PSD matrix.
We can interpret $\cL$ as a \textit{generalized graph Laplacian} matrix corresponding to an irreducible\footnote{An irreducible graph $\cG$ means that there exists a path from any node in $\cG$ to any other node in $\cG$~\cite{milgram1972}.} signed graph $\cG$ with self-loops\footnote{Given a generalized graph Laplacian $\cL$, one can compute the corresponding combinatorial graph Laplacian matrix $\L$ by removing self-loops so that $\L(i,i)~=~-\sum_{j\neq i}\cL(i,j)$ and $\L(i,j)=\cL(i,j)$.}. 
This means that the Gershgorin disc left-ends of $\cL$ are not initially aligned at the same value, and hence GDAS as discussed cannot be used directly. 

However, it has been proven recently in a theorem in \cite{yang2020} (called \textit{Gershgorin disc perfect alignment} (GDPA) in the sequel) that Gershgorin disc left-ends of a generalized graph Laplacian matrix $\cL_{B}$ corresponding to an irreducible, \textit{balanced}, signed graph $\cG_{B}$ can be aligned \textit{exactly} at $\lambda_{\min}(\cL_{B})$ via similarity transform $\S \cL_{B} \S^{-1}$, where $\S = \text{diag}(1/v_1, \ldots, 1/v_{3n})$, and $\v =[v_1\ldots v_{3n}]^{\top}$ is the first eigenvector of $\cL_{B}$ corresponding to the smallest eigenvalue $\lambda_{\min}(\cL_{B})$. First eigenvector $\v$ of $\cL_B$ can be computed efficiently using known fast algorithms, such as \textit{Locally Optimal Block Preconditioned Conjugate Gradient} (LOBPCG)~\cite{knyazev2001}, in roughly linear time, since $\cL_B$ is sparse and symmetric.
Leveraging on GDPA, we perform the following three steps to solve \eqref{eq:obj} approximately. \\
\noindent
\textbf{Step 1:} 
Approximate $\cG = (\cV, \cE, \cU)$ with a balanced graph $\cG_{B}~=~(\mathcal{V}, \mathcal{E}_{B},\mathcal{U})$ while satisfying the following condition:
\begin{equation}
 \lambda_{\min}(\H^{\top}\H+\mu\cL_{B})\leq \lambda_{\min}(\H^{\top}\H+\mu\cL),
 \label{eq:cond1}
 \end{equation}
where $\cL_B$ is the generalized graph Laplacian matrix corresponding to $\cG_B$.\\
\noindent
\textbf{Step 2:} 
Given $\cL_B$, perform similarity transform $\cL_p~=~\S \cL_B \S^{-1}$, where $\S~=~\text{diag}(1/v_1, \ldots, 1/v_{3n})$ and $\v$ is the first eigenvector of $\cL_B$, so that disc left-ends of matrix $\cL_p$ are aligned exactly at  $\lambda_{\min}(\cL_p)~=~\lambda_{\min}(\cL_B)$. \\
\noindent
\textbf{Step 3}:
Perform GDAS~\cite{yuanchao2019ICASSP, bai2020} to maximize $\lambda^-_{\min}(\H^{\top} \H + \mu \cL_p) = \lambda^-_{\min}(\H^{\top} \H + \mu \cL_B)$. 

It is clear that we must obtain a balanced graph $\cG_{B}$ that induces a \textit{tight} lower bound $\lambda_{\min}(\H^{\top}\H+\mu\cL_{B})$ for $\lambda_{\min}(\H^{\top}\H+\mu\cL)$. 
Thus, in next two sections, we define an optimization objective for graph balancing and propose an optimization algorithm to maximize the defined objective.

\section{Graph Balancing Formulation}
\label{sec:balance}

We first show that if we obtain a balanced graph $\cG_{B}$ such that $\L-\L_{B}$ is a PSD matrix (\ie, $\L-\L_{B}\succeq 0$), then \eqref{eq:cond1} will be satisfied, where $\L$ and $\L_{B}$ are combinatorial graph Laplacian matrices of graphs $\cG$ and $\cG_{B}$ respectively.
We state this formally as follows.

\begin{proposition}
\label{prop:lowerbound}
Given two undirected graphs $\mathcal{G}~=~(\mathcal{V}, \mathcal{E},\mathcal{U})$ and $\mathcal{G}_{B}~=~(\mathcal{V}, \mathcal{E}_{B},\mathcal{U})$, if $\L-\L_{B}\succeq 0$, then \eqref{eq:cond1} is satisfied.
\end{proposition}

\begin{proof}
Since $\L-\L_{B}\succeq 0$, we can write
\begin{equation}
    \x^{\top}\L\x\geq\x^{\top}\L_{B}\x,
    \hspace{10pt} \forall \x\in\mathbb{R}^{3n} \hspace{10pt}
    \label{eq:eqlGLR}
\end{equation}
Moreover, since the two graphs have the same set of self-loops $\mathcal{U}$, $\mathrm{diag}(\W)~=~\mathrm{diag}(\W_{B})$, where $\W$ and $\W_{B}$ are adjacency matrices of graphs $\cG$ and $\cG_{B}$ respectively. 
Then, due to \eqref{eq:eqlGLR}, we can write that, $\forall \x\in\mathbb{R}^{3n}$,
\begin{equation}
\begin{split}
   \x^{\top}\underbrace{\left(\L+\mathrm{diag}(\W)\right)}_{\cL}\x\geq\x^{\top}\underbrace{\left(\L_{B}+\mathrm{diag}(\W_{B})\right)}_{\mathcal{L}_{B}}\x. 
\end{split} 
\label{eq:eqGGLR}
\end{equation}
Consequently,
\begin{equation}
\begin{split}
   \x^{\top}(\H^{\top}\H+\mu\cL)\x\geq\x^{\top}(\H^{\top}\H+\mu\cL_{B})\x. 
\end{split}
\label{eq:GGLR_H}
\end{equation}

Denote by $\u$ the unit-norm eigenvector corresponding to the smallest eigenvalue of $\H^{\top}\H+\mu\cL$. 
Hence, 
\begin{equation}
 \lambda_{\min}(\H^{\top}\H+\mu\cL)~=~\u^{\top}(\H^{\top}\H+\mu\cL)\u.  \label{eq:eigvalvec} 
\end{equation}
By the \textit{Min-max theorem} \cite{matrixAnalysis}, we can write
\begin{equation}
 \lambda_{\min}(\H^{\top}\H+\mu\cL_{B})=\min_{\x,\norm{\x}_{2}=1}\x^{\top}(\H^{\top}\H+\mu\cL_{B})\x. 
 \label{eq:minmax}
\end{equation}
By substituting $\x=\u$ in \eqref{eq:GGLR_H} and using~(\ref{eq:eigvalvec}) and~(\ref{eq:minmax}), we can write
\begin{equation}
\begin{split}
  \lambda_{\min}(\H^{\top}\H+\mu\cL)&\geq \u^{\top}(\H^{\top}\H+\mu\cL_{B})\u\\
  &\geq \min_{\x,\norm{\x}_{2}=1}\x^{\top}(\H^{\top}\H+\mu\cL_{B})\x\\
  &=\lambda_{\min}(\H^{\top}\H+\mu\cL_{B})
\end{split}
\end{equation}
which concludes the proof.
\end{proof}

From the proof, one can see that $\lambda_{\min}(\H^{\top}\H+\mu\cL_{B})$ forms a tight lower bound for $\lambda_{\min}(\H^{\top}\H+\mu\cL)$ when $\x^{\top}\L_{B}\x$ is a tight lower bound for $\x^{\top}\L\x$, $\forall \x\in\mathbb{R}^{3n}$. 
Thus, we define our objective for a balanced graph $\cG_{B}$ to maximize $\x^{\top}\L_{B}\x$ subject to $\L-\L_{B}\succeq 0$. 
Specifically, we 
model $\x$ 
as a zero-mean 
\textit{Gaussian Markov random field} (GMRF)~\cite{rue2005} 
over $\mathcal{G}$, 
\ie,  
$\mathbf{x} \sim \mathcal{N}(\boldsymbol{0},\mathbf{\Sigma})$, where $\mathbf{\Sigma}$ is the \textit{covariance matrix} and $\mathbf{\Sigma}^{-1} = \cL + \delta \I$ \cite{gadde2015}. 
The \textit{precision matrix} $\mathbf{\Sigma}^{-1}$ is so defined since $\cL$ is a PSD matrix, $\I$ is the identity matrix, and $\delta$ is chosen to be a very small positive number. Hence a proper covariance matrix $\mathbf{\Sigma}$ is induced by inverting $\cL + \delta \I$.

Next, we choose a balanced graph $\cG_{B}$ to maximize the expectation of $\x^{\top}\L_{B}\x$ subject to $\L-\L_{B} \succeq 0$:
\begin{equation}
\begin{split}
    \mathop{\mathbb{E}}[\x^{\top}\L_{B}\x]&=\mathop{\mathbb{E}}[\mathrm{Tr}(\x^{\top}\L_{B}\x)]\\
    &=\mathop{\mathbb{E}}[\mathrm{Tr}(\L_{B}\x\x^{\top})]\\
    &=\mathrm{Tr}(\L_{B}\mathop{\mathbb{E}}[\x\x^{\top}])=\mathrm{Tr}(\L_{B}\mathbf{\Sigma}).
\end{split}
\label{eq:expect}
\end{equation}
Maximizing \eqref{eq:expect} will thus promote tightness of $\x^{\top}\L_{B}\x$ given 
$\mathbf{x} \sim \mathcal{N}(\boldsymbol{0},\mathbf{\Sigma})$. 
We can now formalize our optimization problem for graph balancing as follows:
\begin{equation}
\small
\begin{split}
    &\max_{\L_{B}}\mathrm{Tr}(\L_{B}\mathbf{\Sigma})\hspace{5pt};\hspace{4pt} \text{s.t.} \hspace{4pt} \begin{cases}
    \L-\L_{B}\succeq 0,\\ \L_{B}\in\mathcal{B}\subset \mathbb{R}^{3n\times 3n}
    \end{cases}
\end{split}
\label{eq:balancing_opt_prob}
\end{equation}
where $\mathcal{B}$ is the set of $3n\times 3n$ combinatorial graph Laplacian matrices corresponding to balanced signed graphs.

Like previous graph approximation problems such as~\cite{zeng2017} where the ``closest" bipartite graph is sought to approximate a given non-bipartite graph, the balance graph approximation problem in \eqref{eq:balancing_opt_prob} is difficult because of its combinatorial nature.  
Thus, similar in approach to~\cite{zeng2017},
in the next section we propose a greedy algorithm to add one ``most beneficial" node at a time (with corresponding edges) to construct a balanced graph as solution to the problem in \eqref{eq:balancing_opt_prob}.

\section{Graph Balancing Algorithm}
\label{sec:algo}
\subsection{Notations and Definitions}
To facilitate understanding of our graph balancing algorithm, we first introduce the following notations and definitions. 

\vspace{0.07in}
\noindent
\textbf{1.} We define the notion of \textit{consistent} edges in a balanced graph, stemming from Theorem~\ref{theorem:CHT}, as follows.

\begin{definition}
\label{def:edgestype}
A consistent edge is a positive edge connecting two nodes of the same color, or a negative edge connecting two nodes of opposite colors.
An edge that is not consistent is an inconsistent edge.
\end{definition}
\noindent Given this definition, for an edge  $(i,j)\in\mathcal{E}_{B}$, we can write
\begin{equation}
\beta_{i}\beta_{j}\text{sign}(\W_{B}(i,j))=\begin{cases}
1; \hspace{3pt} \text{if} \hspace{3pt} (i,j) \hspace{3pt} \text{is consistent}\\
-1; \hspace{3pt} \text{if} \hspace{3pt} (i,j) \hspace{3pt} \text{is inconsistent},
\end{cases}
\label{eq:edgeweight_sign}     
\end{equation}
where $\beta_i$ denotes the color of node $i$ (\ie, $\beta_i~=~1$ if node $i$ is blue and $\beta_i~=~-1$ if node $i$ is red).

\vspace{0.07in}
\noindent
\textbf{2.} Given graph $\cG = (\cV, \cE, \cU)$, we define  a \textit{bi-colored} node set $\cS \subseteq \cV$, where all edges connecting nodes in $\cS$ are consistent. 
Further denote by $\cC \subseteq \cV \setminus \cS$ the set of nodes within one hop from $\cS$. 
An example of graph $\cG$ with set $\cC$ and $\cS$ is shown in Fig.\;\ref{fig:S&C}. 
If all nodes of graph $\cG$ are in set $\cS$, \ie, $\cS = \cV$ and $\cC = \emptyset$, then $\cG$ is balanced by Theorem~\ref{theorem:CHT}.    

\vspace{0.07in}
\noindent
\textbf{3.} Using \eqref{eq:edgeweight_sign}, for each node $j\in\cC$ we define two disjoint edge sets for edges connecting $j$ to $\cS$: 
i) the set $\cF_j$ of consistent edges from $j$ to $\cS$ when $\beta_{j}=1$, 
ii) the set $\cH_j$ of consistent edges from $j$ to $\cS$ when $\beta_{j}=-1$. 
Note that $\cF_j$ ($\cH_j$) is also the set of inconsistent edges when $\beta_j=-1$ ($\beta_j=1$). 
As an illustration, an example of edge sets $\cF_j$ and $\cH_j$ connecting $j\in\cC$ to $\cS$ is shown in Fig.\;\ref{fig:const_edge}, where consistent and inconsistent edges are drawn in different colors for a given value $\beta_{j}$. 
We see that if we remove inconsistent edges from $j\in\cC$ to $\cS$ for a given $\beta_{j}$ value, node $j$ can be added to $\cS$.

\vspace{0.07in}
\noindent
\textbf{4.} Denote by $f_{j}(\beta_{j})$ the value of the objective in \eqref{eq:balancing_opt_prob} when node $j$ is assigned color $\beta_j$ and the corresponding inconsistent edges from $j\in \cC$ to $\cS$ are removed. 
This means removing $\cH_j$ if $\beta_j=1$, and removing $\cF_j$ if $\beta_j=-1$.

\subsection{Balancing Algorithm}
\label{sec:bal_algorithm}

Using the above definitions, we present our iterative greedy algorithm to approximately solve \eqref{eq:balancing_opt_prob}.
In a nutshell, given an unbalanced graph $\cG$, we construct a corresponding balanced graph $\cG_B$ by adding one node at a time to set $\cS$: 
at each iteration, we select a ``most beneficial" node $j \in \cC$ with color $\beta_j$ to add to $\cS$, so that the objective in \eqref{eq:balancing_opt_prob} is locally maximized. 

Specifically, first, we initialize bi-colored set $\cS$ with a random node $i$ and color it blue, \ie, $\beta_{i}=1$.
Then, at a given iteration, we select a node $j \in \cC$ and corresponding color $\beta_j$ that maximizes objective $f_j(\beta_j)$, \ie, 
\begin{equation}
(j^{*}, \beta_{j^*}^{*})=\arg\max_{j \in \cC,\beta_j \in \{-1, 1\}} f_j(\beta_j).
\label{eq:locallymax}
\end{equation}
We add node $j^*$ of solution $(j^*,\beta_{j^*}^*)$ in \eqref{eq:locallymax} to set $\cS$, assign it color $\beta_{j^*}^*$, and remove inconsistent edges from $j^*$ to $\cS$ to conclude the iteration. 
Main steps of the algorithm are summarized as follows:\\

\noindent\textbf{Step 1:} Initialize set $\cS$ with a random node $i$ and set $\beta_i = 1$.\\
\noindent \textbf{Step 2:} Compute $f_{j}(1)$ and $f_{j}(-1)$ for each $j\in\cC$.\\
\noindent\textbf{Step 3:} Compute solution $(j^{*},\beta^*_{j^*})$ using~(\ref{eq:locallymax}).\\   
\noindent\textbf{Step 4:} Remove all inconsistent edges from $j^{*}\in\cC$ to $\cS$.\\
\noindent\textbf{Step 5:} $\mathcal{S}\leftarrow \mathcal{S}\cup j^{*} $. \\
\noindent\textbf{Step 6:} Update $\cC$ according to modified $\cS$.\\
\noindent\textbf{Step 7:} Repeat steps 2-6 until $\cC = \emptyset$.\\

\begin{figure}[t]
\centering
\includegraphics[width=0.35\textwidth]{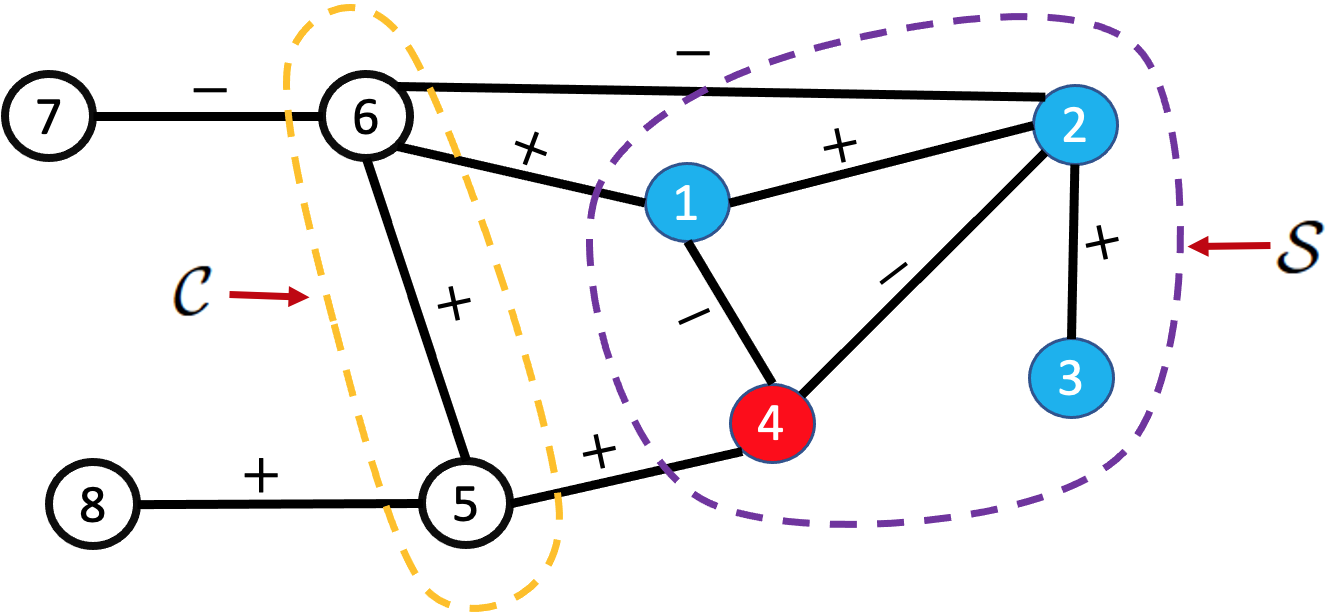}
\vspace{-11pt}
\caption{An example of an 8-node graph $\cG$ with sets $\cS$ and $\cC$.}
\label{fig:S&C}
\end{figure}

\begin{figure}[t]
\centering
\includegraphics[width=0.46\textwidth]{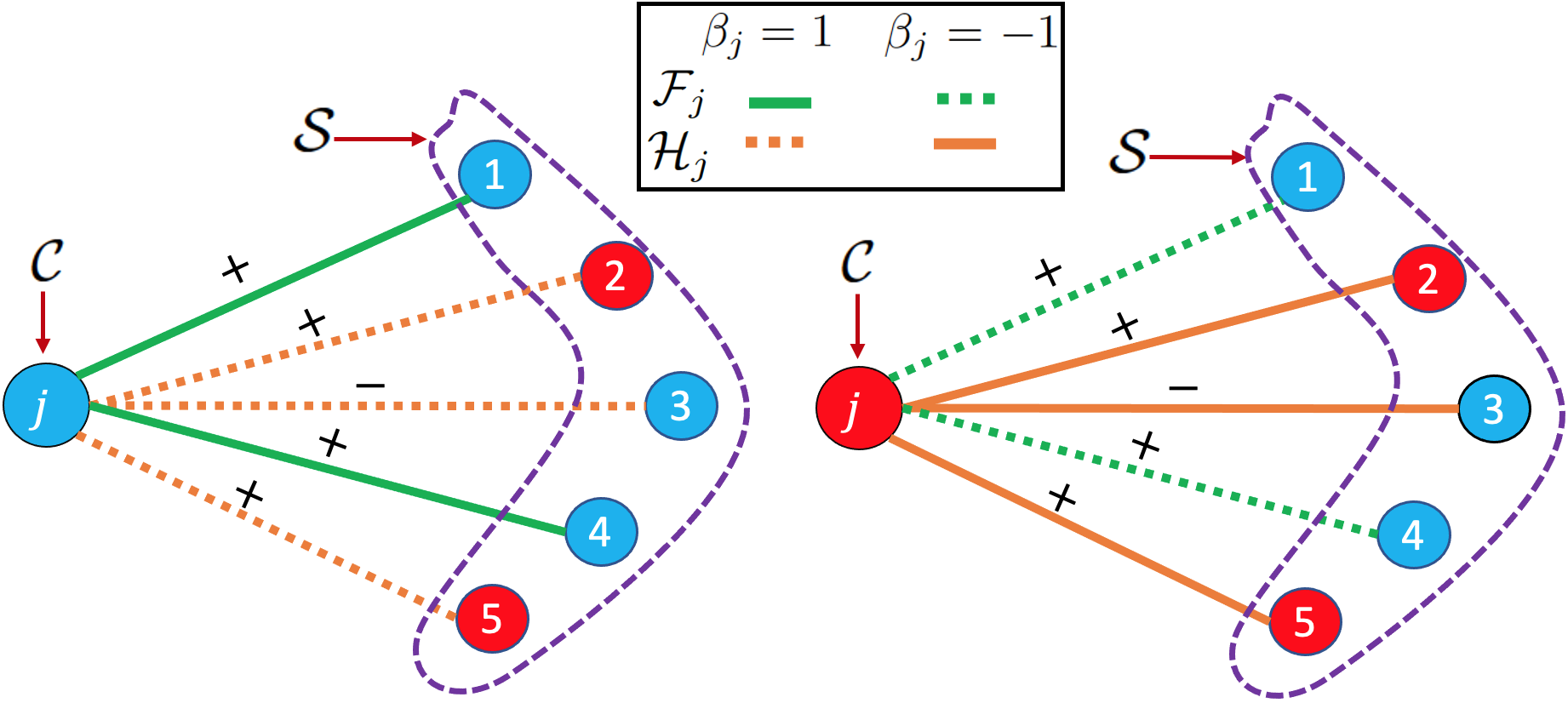}
\vspace{-10pt}
\caption{An example of consistent edges $\cF_j$ and inconsistent edges $\cH_j$ connecting node $j\in\cC$ to $\cS$ when $\beta_{j}=1$ (left), and consistent edges $\cH_j$ and inconsistent edges $\cF_j$ when $\beta_{j}-1$ (right). 
For simplicity, we omit edges connecting nodes within set $\cS$.}
\label{fig:const_edge}
\vspace{-13pt}
\end{figure}

\subsection{Inconsistent Edge Removal}
When removing inconsistent edges in step 2 (to compute $f_{j}(1)$ and $f_{j}(-1)$) and step 4, the constraint $\L-\L_{B}\succeq 0$ in optimization \eqref{eq:balancing_opt_prob} must be preserved. 
One can show that when removing a positive edge, the constraint $\L-\L_{B}\succeq 0$ is always maintained as formally stated in Proposition~\ref{prop:positive}.

\begin{proposition}
\label{prop:positive}
Given an undirected graph $\mathcal{G}~=~(\mathcal{V}, \mathcal{E},\mathcal{U})$, removing a positive edge $(q,r) \in \mathcal{E}$ from $\mathcal{G}$---resulting in graph $\cG_B = \left(\cV, \mathcal{E}\setminus(q,r),\mathcal{U} \right)$---entails $\L-\L_{B} \succeq 0$. 
\end{proposition}
\begin{proof}
By \eqref{eq:GLR}, for graph $\cG$, we can write following expressions $\forall \x\in\mathbb{R}^{3n}$,
\begin{equation}
\begin{split}
    \x^{\top}\L\x&=\sum_{i,j\in \mathcal{E}}w_{i,j}(x_{i}-x_{j})^{2},  \\
    &=w_{q,r}(x_{q}-x_{r})^{2}+\sum_{i,j\in \mathcal{E}\setminus(q,r)}w_{i,j}(x_{i}-x_{j})^{2}.
\end{split}
\label{eq:postive_edge_remove}
\end{equation}
Since $w_{q,r}>0$, we can write following inequality using~(\ref{eq:postive_edge_remove}),
\begin{equation}
    \x^{\top}\L\x \geq \sum_{i,j\in \mathcal{E}\setminus(q,r)}w_{i,j}(x_{i}-x_{j})^{2}=\x^{\top}\L_{B}\x.
\end{equation}
Hence, $\L-\L_{B}\succcurlyeq 0$, which concludes the proof
\end{proof}
\noindent Given the ease of removing inconsistent positive edges, we first remove them all before removing inconsistent negative edges.

In contrast, removal of a negative edge can violate $\L-\L_{B}\succeq 0$.  
Thus, when removing a negative inconsistent edge $(j,i)$, where $j\in\cC$ and $i\in\cS$, we update weights of \textit{two} additional edges that together with $(j,i)$ form a triangle, in order to ensure constraint $\L-\L_{B} \succeq 0$ is satisfied, while keeping updated edges consistent.
We state this formally in Proposition \ref{prop:negative egdes}.

\begin{proposition}
\label{prop:negative egdes}

Given an undirected graph $\mathcal{G}~=~(\mathcal{V}, \mathcal{E},\mathcal{U})$, assume that there is an inconsistent negative edge $(j,i)\in \mathcal{E}$ connecting nodes $j\in\cC$ and $i\in\cS$ of the same color.
Let $k \in \cS$ be a node of opposite color to nodes $j$ and $i$.
If edge $(j,i)$ is removed---resulting in graph $\mathcal{G}_{B}=(\mathcal{V}, \mathcal{E} \setminus (j,i), \mathcal{U})$---and weights for edges $(k,j)$ and $(k,i)$ are updated as
\begin{equation}
\tilde{w}_{p,q}=
w_{p,q}+2w_{j,i} ~~~ \mbox{for} ~ (p,q) \in\{(k,j), (k,i)\} 
\label{eq:weightupdate}
\end{equation}
then a) $\L-\L_{B}\succeq 0$; and b) $(k,j)$ and $(k,i)$ are consistent edges.
\end{proposition}

The proof for part a) of  Proposition~\ref{prop:negative egdes} is reported in Appendix~B in the supplementary material; the proof for part b) is given below.
\begin{proof}
Since, by assumption, node $k \in \cS$ is of opposite color to $i \in \cS$, if edge $(k,i) \in \cE$, then $w_{k,i} < 0$ for $(k,i)$ to be consistent.
Similarly, if $(k,j) \in \cE$, then $w_{k,j} < 0$, because by assumption $k$ and $j$ are of opposite colors, and inconsistent positive edges have already been removed. Hence $(k,j)$ is also a consistent edge.
This means that edge weight update in \eqref{eq:weightupdate} will only make $w_{k,i}$ and $w_{k,j}$ more negative, and would not switch edge sign and affect the consistency of edges $(k,i)$ and $(k,j)$. 
\end{proof}

We use a simple example in Fig.~\ref{fig:weight_update_case1} to illustrate how a negative inconsistent edge is removed. 
When removing a negative inconsistent edge $(j,i)$ from Fig.~\ref{fig:weight_update_case1} (left), we select $k\in \cS$, where $i,k\in \cS$ are of opposite colors. 
If negative edges $(k,i)$ and/or $(k,j)$ do not exist, we create edges with zero weight $w_{k,i}=0$ and/or $w_{k,j}=0$, resulting in a triangle $(i,k,j)$. 

Then, using Proposition~\ref{prop:negative egdes}, we remove negative edge $(j,i)$ and update edge weights $w_{i,k}$ and $w_{k,j}$ as shown in Fig.~\ref{fig:weight_update_case1} (right). 
We see that the updated edges $(k,i)$ and $(k,j)$ maintain the same signs and are consistent. 
Further, by Proposition~\ref{prop:negative egdes}, this edge weight update ensures that constraint $\L-\L_{B}\succeq 0$ remains satisfied.

Denote by $\cK$ the set of candidates for $k \in \cS$. 
One approach is to choose $k$ within candidate solution set $\cK$ so that $f_{j}^{k}(\beta_{j})$ (\ie,  value of $f_{j}(\beta_{j})$ for a given $k\in\mathcal{K}$) is maximized, \ie,
\begin{equation}
    k^{*}=\arg\max_{k\in\mathcal{K}}f_{j}^{k}(\beta_{j}).
\end{equation}
Such an exhaustive search for each $k\in\mathcal{K}$ for large $|\mathcal{K}|$ would increase the computational complexity significantly.
Instead, if the candidate set  $|\mathcal{K}|$ is greater than 1, we select a $k\in\mathcal{K}$ randomly. 
Experimental results show that there is no significant performance difference between the exhaustive approach and the random approach.

For a given value of $\beta_{j}$, if $|\mathcal{K}|=0$, then it means that there are no nodes in $\cS$ of opposite color to $i$, or more simply, all nodes in $\cS$ are of the same color.
In this case, it suffices to color $j \in \cC$ to be the opposite color to nodes in $\cS$ and remove all inconsistent positive edges to $\cS$.

\begin{figure}[t]
\centering
\includegraphics[width=0.38\textwidth]{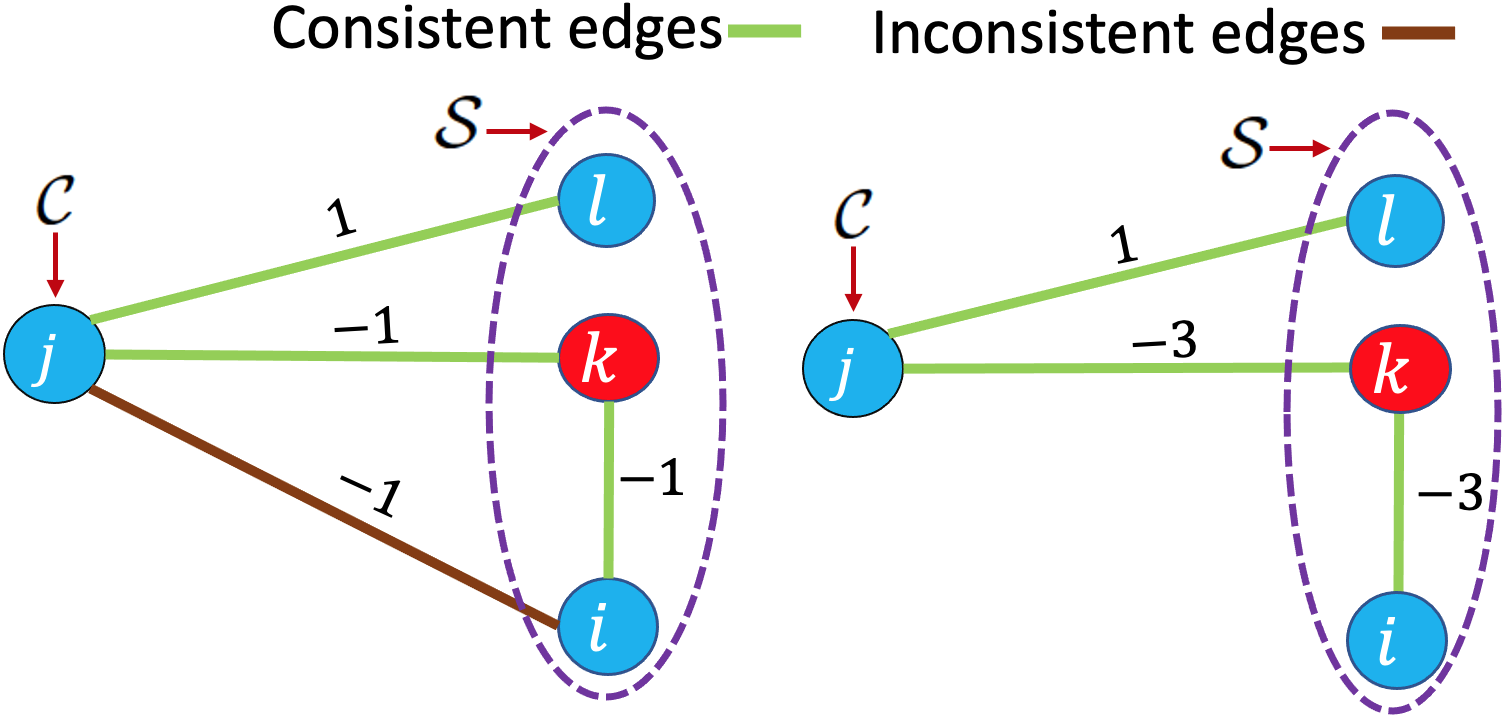}
\vspace{-8pt}
\caption{An example for edge weight updating. A cycle $i \rightarrow j \rightarrow k$ with inconsistent negative edge $(j,i)$ need to be removed (left); updated edge weights after removing the negative inconsistent edge (right).}
\label{fig:weight_update_case1}
\end{figure}

\section{Experimental Results}
\label{sec:exp}

\begin{figure}[t]
\centering
\includegraphics[width=0.48\textwidth]{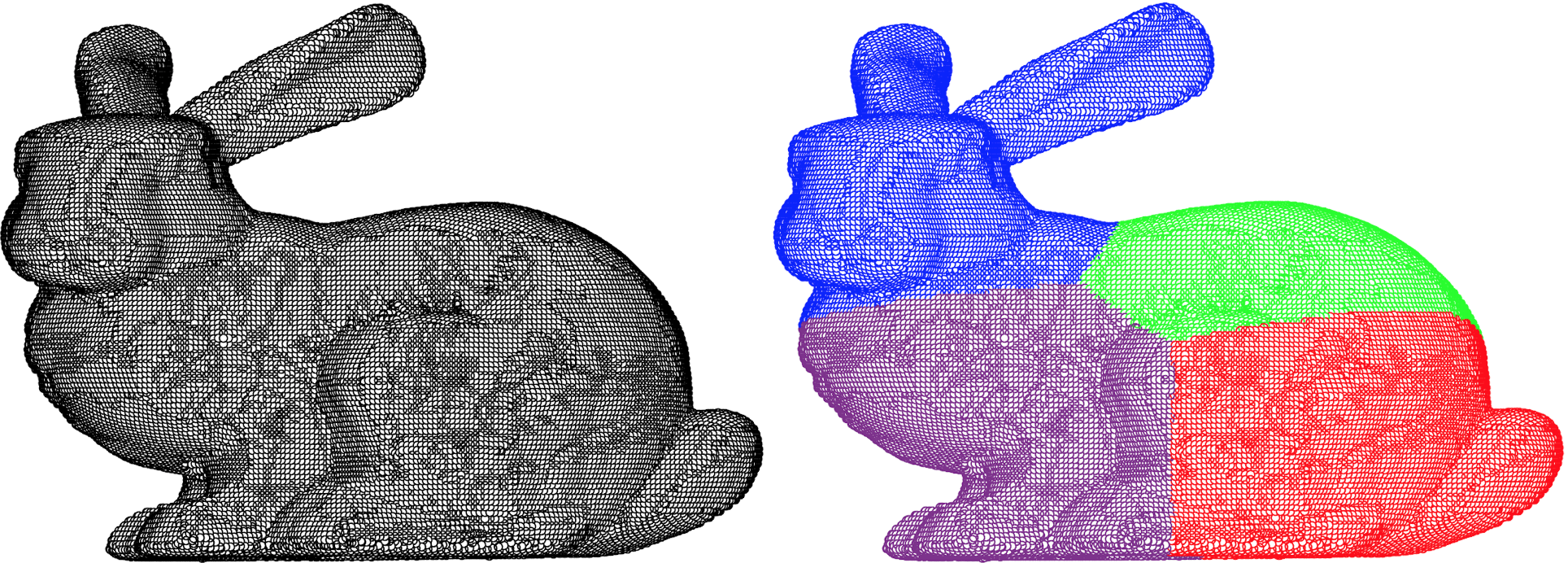}
\label{fig:pcd-day}
\vspace{0pt}
\caption{An example of a PC model (left) and its associated sub-clouds (right). Each color represents a different sub-cloud.} 
\label{fig:clust}
\vspace{0pt}
\end{figure}

\begin{table*}[t]
\centering 
\caption{Average SR reconstruction $\text{C2C}(\times 10^{-2})$ error per PC model for different sub-sampling algorithms using different PC SR methods.}
\vspace{-8pt}
\resizebox{\textwidth}{!}{
\begin{tabular}{c|c|c|c|c|c||c|c|c|c|c||c|c|c|c|c||c|c|c|c|c}
\cline{2-21}
\multicolumn{1}{l|}{}                                                            & \multicolumn{20}{c}{Different SR methods under five sub-sampling ratios: $0.2|0.3|0.4|0.5|0.6$}                                                                          \\ \hline
\multicolumn{1}{c|}{\begin{tabular}[c]{@{}c@{}}Sub-sampling\\ methods\end{tabular}} & \multicolumn{5}{c||}{APSS}      & \multicolumn{5}{c||}{EAR}       & \multicolumn{5}{c||}{GTV}         & \multicolumn{5}{c}{FGTV}        \\ \hline
\multicolumn{1}{c|}{STS}                                                         & 5.90 & 4.52  & 3.85  & 3.32  & 3.08  & 5.24 & 4.07 & 3.42 & 2.98 & 2.87 & 5.03 & 3.82  & 3.14 & 2.85  & 2.59  & 4.26 & 2.85  & 2.76  & 2.47  & 2.29  \\ \hline
\multicolumn{1}{c|}{MCS}                                                        & 5.41     & 4.27      & 3.75      & 3.24      & 2.98 & 5.11 & 3.82      &3.21      & 2.95      & 2.61     &    4.97  & 3.52     &  3.06    & 2.83     & 2.52     &     4.29 & 3.10      & 2.77      & 2.48      & 2.30      \\ \hline
\multicolumn{1}{c|}{MESS}                                                        &  4.21    & 2.95     &    2.37  & 1.97      & 1.72 &  3.75  &  2.51     & 1.98      &  1.63    &1.41      &      3.68& 2.40      & 1.95      & 1.52      &  1.33     &      3.44& 2.11     &1.76      & 1.46     & 1.28     \\ \hline
\multicolumn{1}{c|}{PDS}                                                        &  3.12    & 2.41      & 2.02      & 1.83      & 1.72 &   2.64   & 1.86      & 1.52      & 1.40      & 1.35      &      2.58 &  1.75    & 1.47     &    1.38  & 1.33      &     2.51 & 1.60      & 1.45      & 1.36      & 1.29      \\ \hline
\multicolumn{1}{c|}{RS}                                                       &   4.58   &  4.50    & 2.98      & 2.67     & 2.46      &  4.45   & 3.92      & 2.85      & 2.50      & 2.35      &     4.35 & 3.36      & 2.61      & 2.18      & 2.01      &     3.93 & 3.09      & 2.23      & 2.09      & 1.85     \\ \hline
\multicolumn{1}{c|}{BGFS}                                                        &   4.48   & 3.75      & 2.99      & 2.58      & 2.43  &   4.36   & 3.63      & 2.84      & 2.52      & 2.31      &   4.28   & 3.30      & 2.54      & 2.10      & 1.95      &     3.76 & 3.05      & 2.17      & 2.04      & 1.78      \\ \hline
\multicolumn{1}{c|}{CS}                                                        &  4.53    & 3.88      & 2.94      & 2.50     & 2.41 &   4.42   & 3.84      & 2.75     & 2.45      &  2.34      &    4.31  & 3.25      & 2.48      & 2.07     &  1.82      &     3.95  &   3.02      &  2.11     & 1.92     &   1.63   \\ \hline
\multicolumn{1}{c|}{PS}                                                        &  4.43    & 3.80      & 2.87      & 2.45      & 2.23     &   4.05   & 3.24      & 2.45      & 2.08      & 1.92      &   3.91   & 3.07      &  1.98    & 1.90      &  1.79    &    3.72  & 2.94     & 1.95      &  1.87    & 1.75     \\ \hline
\multicolumn{1}{c|}{FPS}                                                       &   3.01   & 2.22      & 1.87      & 1.56      & 1.41  &   2.87   & 2.03     & 1.54      & 1.42      & 1.34      &   2.62   &  1.80    & 1.41      & 1.35      & 1.26      &     2.59 & 1.71     &  1.40    & 1.30     & 1.24     \\ \hline
\multicolumn{1}{c|}{AGBS}                                                   &  2.75    & 1.93      & 1.67      & 1.45      &  1.41  &  2.67    & 1.77      & 1.51      & 1.34      & 1.23      &    2.58  & 1.75     & 1.46      & 1.28      & 1.18     &     2.50 & 1.53     & 1.37     & 1.25     & 1.16     \\ \hline
\multicolumn{1}{c|}{Proposed}                                                   & \textbf{2.48}     & \textbf{1.70}      & \textbf{1.37}      & \textbf{1.30}      & \textbf{1.28}     &   \textbf{2.34}   & \textbf{1.62}      & \textbf{1.31}      & \textbf{1.25}      & \textbf{1.16}      &   \textbf{2.29}   & \textbf{1.53}      & \textbf{1.27}      & \textbf{1.18}      & \textbf{1.07}      & \textbf{2.21} & \textbf{1.42}     & \textbf{1.24}     &  \textbf{1.12}    & \textbf{1.03}     \\ \hline
\end{tabular}
}
\label{tab:C2C}
\vspace{-7pt}
\end{table*}

\begin{table*}[t]
\centering 
\caption{Average SR reconstruction $\text{C2P}(\times 10^{-3})$ error per PC model for different sub-sampling algorithms using different PC SR methods.}
\vspace{-8pt}
\resizebox{\textwidth}{!}{\begin{tabular}{c|c|c|c|c|c||c|c|c|c|c||c|c|c|c|c||c|c|c|c|c}
\cline{2-21}
\multicolumn{1}{l|}{}                                                            & \multicolumn{20}{c}{Different SR methods under five sub-sampling ratios: $0.2|0.3|0.4|0.5|0.6$}                                                                          \\ \hline
\multicolumn{1}{c|}{\begin{tabular}[c]{@{}c@{}}Sub-sampling\\ methods\end{tabular}} & \multicolumn{5}{c||}{APSS}       & \multicolumn{5}{c||}{EAR}       & \multicolumn{5}{c||}{GTV}         & \multicolumn{5}{c}{FGTV}        \\ \hline
\multicolumn{1}{c|}{STS}                                                         & 6.67 & 3.84 & 3.33 & 3.05 & 2.78  & 6.03 & 3.14 & 2.93  & 2.60  & 2.47  & 5.92 & 2.94 & 2.79 & 2.40 & 2.28 & 5.67 & 2.74  & 2.65  & 2.27  & 2.02  \\ \hline
\multicolumn{1}{c|}{MCS}                                                        &  6.80    & 4.43      & 3.47      & 2.94      & 2.55      &     6.15 & 4.22      & 3.08      & 2.57      & 2.34      &  6.02    &  4.08     & 2.90      & 2.46      &   2.22 &    5.65  & 3.40      & 2.77      & 2.28      & 2.03      \\ \hline
\multicolumn{1}{c|}{MESS}                                                        &    4.91  & 2.66      & 2.02      & 1.37      & 0.91      &   4.30   & 2.25      & 1.61      & 1.09     & 0.73      &   4.11   & 2.01      & 1.44      & 0.90      &  0.58      &   3.90   & 1.80      & 1.24     & 0.78      & 0.50      \\ \hline
\multicolumn{1}{c|}{PDS}                                                        &  3.03    & 1.87      & 1.34      & 1.05      & 0.76     &   2.56   & 1.43      & 1.03      & 0.75      & 0.58      &   2.38   & 1.27      & 0.87      & 0.59      & 0.46      &  2.20    & 1.11      & 0.74     & 0.50      & 0.36      \\ \hline
\multicolumn{1}{c|}{RS}                                                       &  5.07    & 2.95      & 2.17      & 1.54      & 1.05   &   4.38   & 2.31      & 1.75      & 1.09      & 0.73      &  4.14    & 2.10      & 1.57      & 0.98      & 0.62      &     3.98 & 1.97      & 1.52      & 0.90      & 0.52      \\ \hline
\multicolumn{1}{c|}{BGFS}                                                        &  4.81    & 2.90      & 2.08      & 1.50      & 0.97  &   4.32   & 2.24      & 1.68      & 1.04      & 0.70      &  4.05    & 2.05      & 1.55      & 0.94     & 0.58      &    3.75  & 1.84      & 1.45      & 0.85      & 0.49      \\ \hline
\multicolumn{1}{c|}{CS}                                                        &  5.01    & 2.84      & 2.11      & 1.44     & 0.98 &  4.27    & 2.18     & 1.63      & 1.05      & 0.68      &  4.10    & 2.03      & 1.48      & 0.91      & 0.55      &  3.90    & 1.87      & 1.36      & 0.83      & 0.48      \\ \hline
\multicolumn{1}{c|}{PS}                                                        &   4.86   & 2.51      & 1.82      & 1.14      & 0.88     &   4.18   & 2.02      & 1.40     & 0.86      & 0.63      &   4.05   & 1.87      &1.28      & 0.77      & 0.52     &  3.85    &  1.72      & 1.14      & 0.67      & 0.46      \\ \hline
\multicolumn{1}{c|}{FPS}                                                       &  3.72    & 1.98      & 1.38      &  0.97    & 0.60 &     3.12 & 1.54     & 0.95     & 0.60      &  0.39  &  2.93    & 1.42     & 0.81      & 0.48      & 0.29      &   2.77   & 1.30      & 0.70      & 0.40      & 0.24      \\ \hline
\multicolumn{1}{c|}{AGBS}                                                   &  3.24    & 1.78      & 1.45      & 1.08      & 0.68  &   2.67   & 1.33      & 1.07      & 0.76      & 0.53      &  2.54    & 1.19     & 0.93      & 0.65      & 0.42      &   2.40   & 1.07      & 0.84      & 0.57      & 0.35      \\ \hline
\multicolumn{1}{c|}{Proposed}                                                   &  \textbf{2.87}    & \textbf{1.60}      & \textbf{1.18}      & \textbf{0.88}      &    \textbf{0.54} &   \textbf{2.29}   & \textbf{1.24}      & \textbf{0.88}     & \textbf{0.56}     & \textbf{0.37}      &   \textbf{2.18}   & \textbf{1.05}      & \textbf{0.72}     & \textbf{0.45}      & \textbf{0.28}      &    \textbf{1.98}  & \textbf{0.94}      & \textbf{0.61}      & \textbf{0.37}      & \textbf{0.22}      \\ \hline
\end{tabular}}
\label{tab:C2P}
\end{table*}

\begin{figure*}[t]
\centering
\subfloat[BGFS]{
\includegraphics[width=0.25\textwidth]{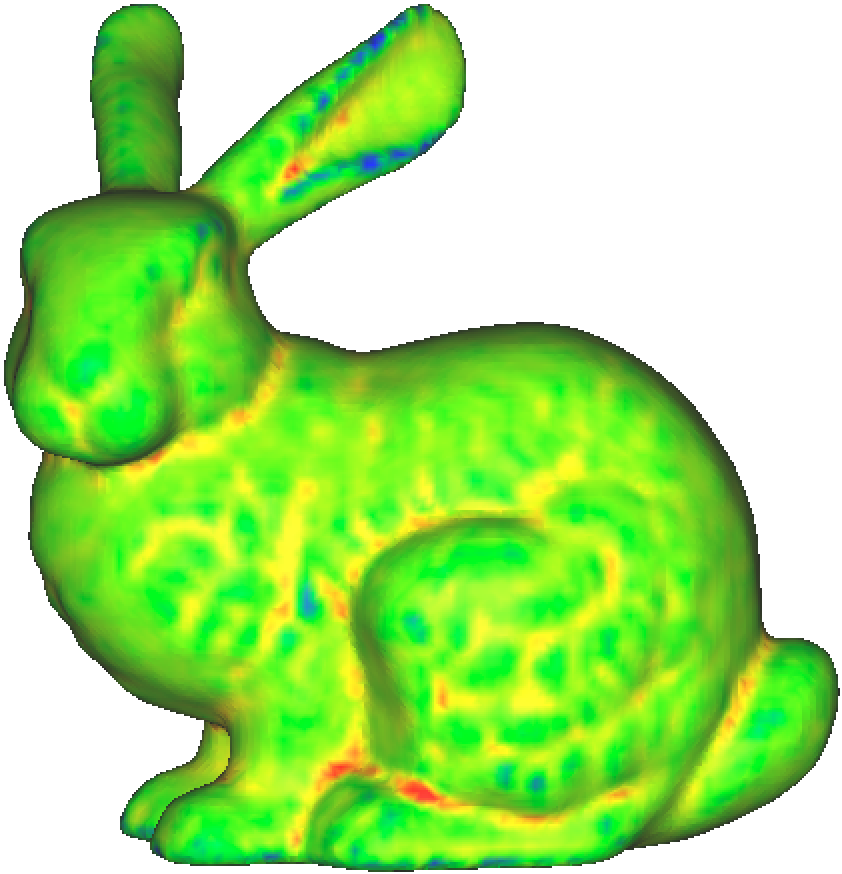}
\label{fig:pcd-on}}
\hspace{-8pt}
\subfloat[{PDS}]{
\includegraphics[width=0.25\textwidth]{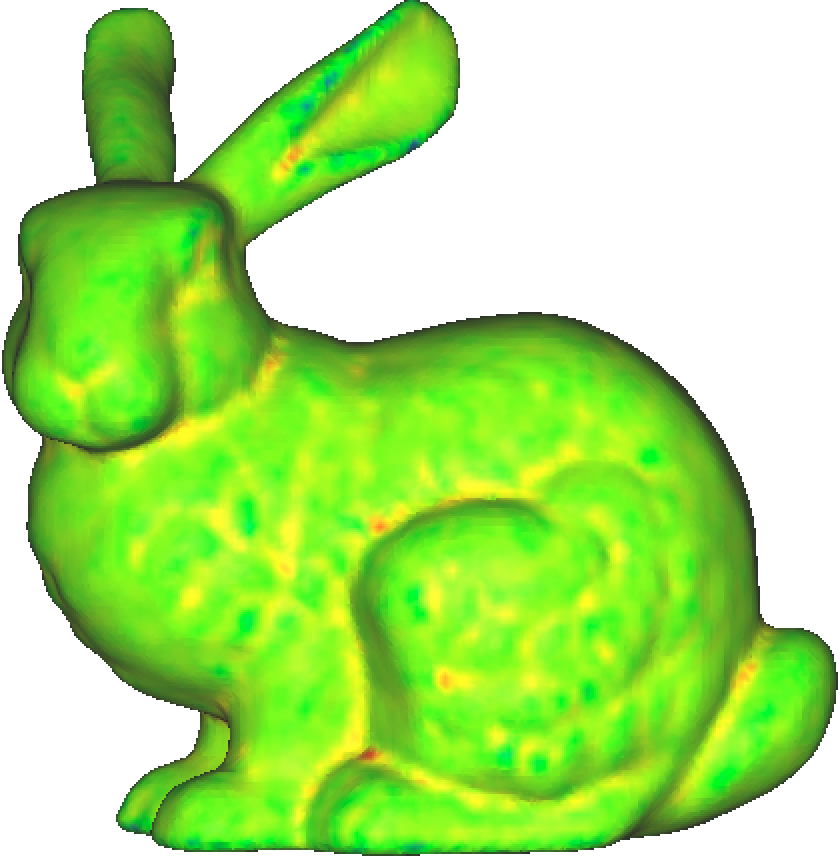}
\label{fig:pcd-off}}
\hspace{-8pt}
\subfloat[proposed]{
\includegraphics[width=0.25\textwidth]{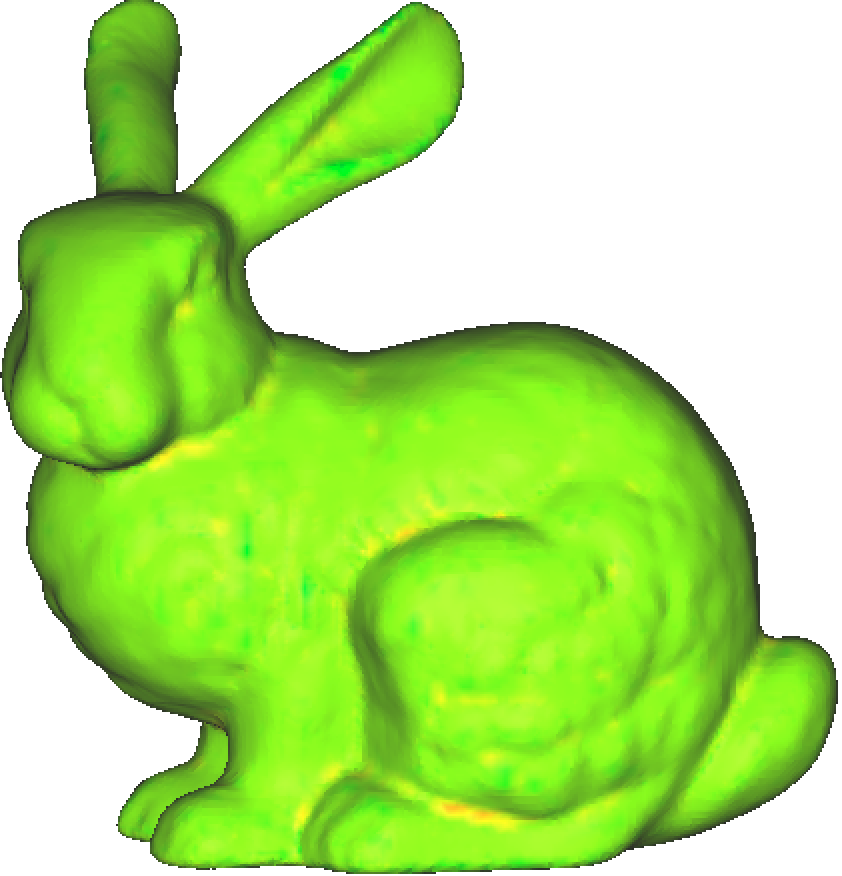}
\label{fig:tvlcd-ofn}}
\subfloat{
\includegraphics[width=0.04\textwidth]{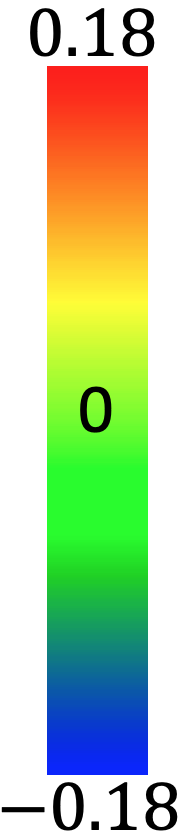}
\label{fig:tvlcd-ofn}}
\caption{SR reconstruction results obtained using FGTV SR method from different methods of sub-sampled \texttt{Bunny} models under 0.2 sub-sampling ratio. Here the surfaces are colorized by the distance from the ground truth surface (color-map is included).} 
\label{fig:bunny}
\vspace{-14pt}
\end{figure*}

\begin{figure*}[!htb]
\centering
\vspace{-0pt}
\subfloat[BGFS]{
\includegraphics[width=0.25\textwidth]{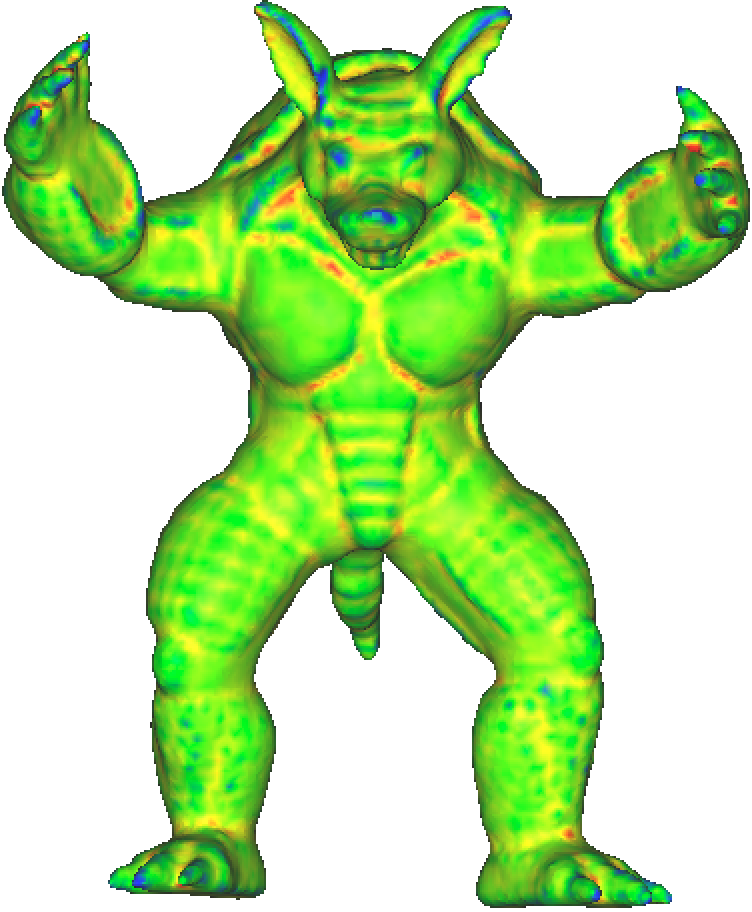}
\label{fig:pcd-on}}
\subfloat[{PDS}]{
\includegraphics[width=0.25\textwidth]{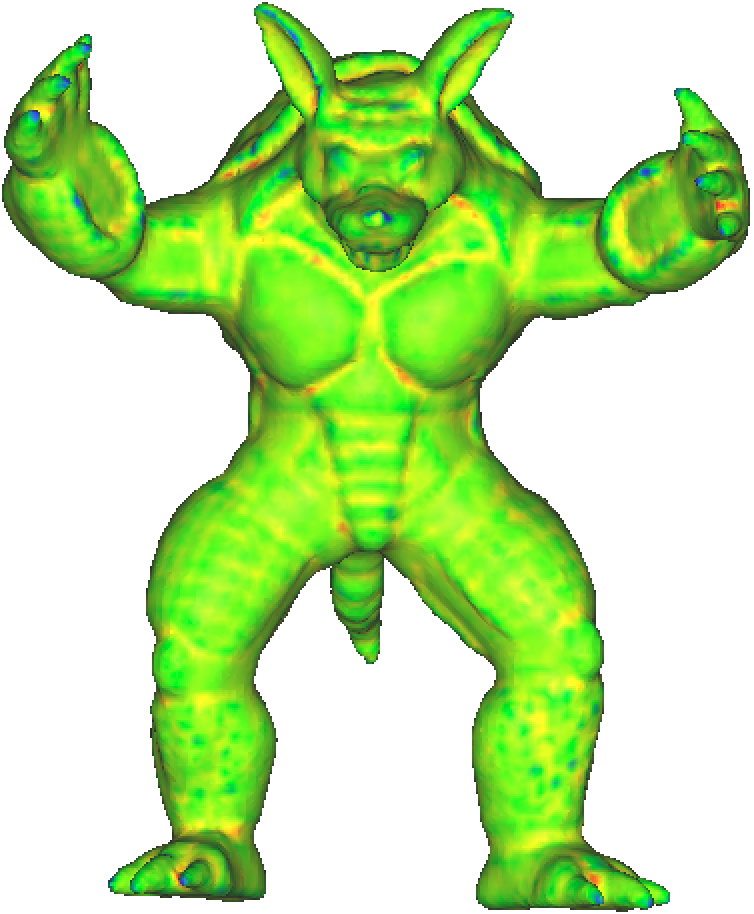}
\label{fig:pcd-off}}
\subfloat[proposed]{
\includegraphics[width=0.25\textwidth]{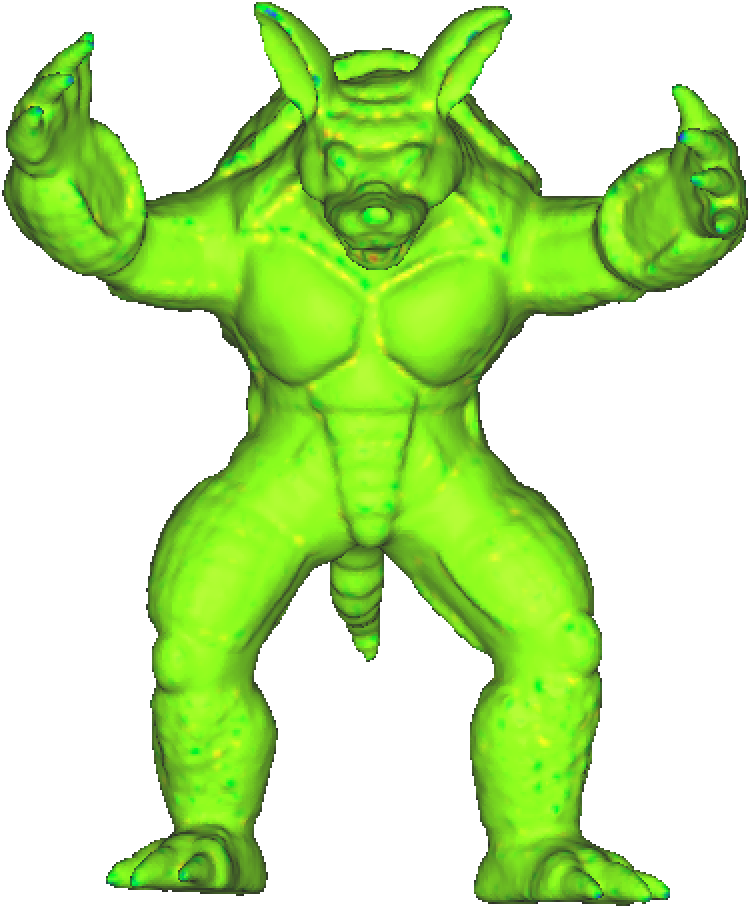}
\label{fig:tvlcd-ofn}}
\subfloat{
\includegraphics[width=0.04\textwidth]{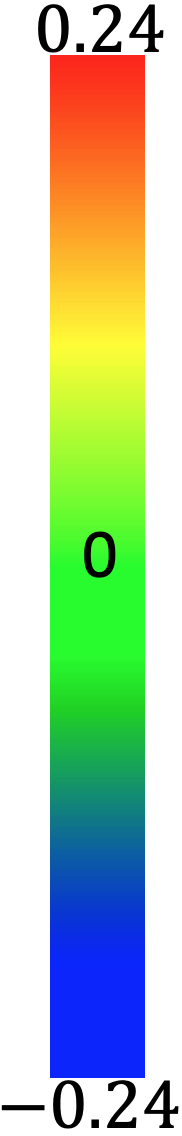}
\label{fig:tvlcd-ofn}}
\caption{SR reconstruction results obtained using GTV SR method from different methods of sub-sampled \texttt{Armadillo} models under 0.2 sub-sampling ratio. Here the surfaces are colorized by the distance from the ground truth surface (color-map is included)} 
\label{fig:armadillo}
\vspace{-10pt}
\end{figure*}

\begin{figure*}[!htb]
\centering
\vspace{-0pt}
\subfloat[BGFS]{
\includegraphics[width=0.25\textwidth]{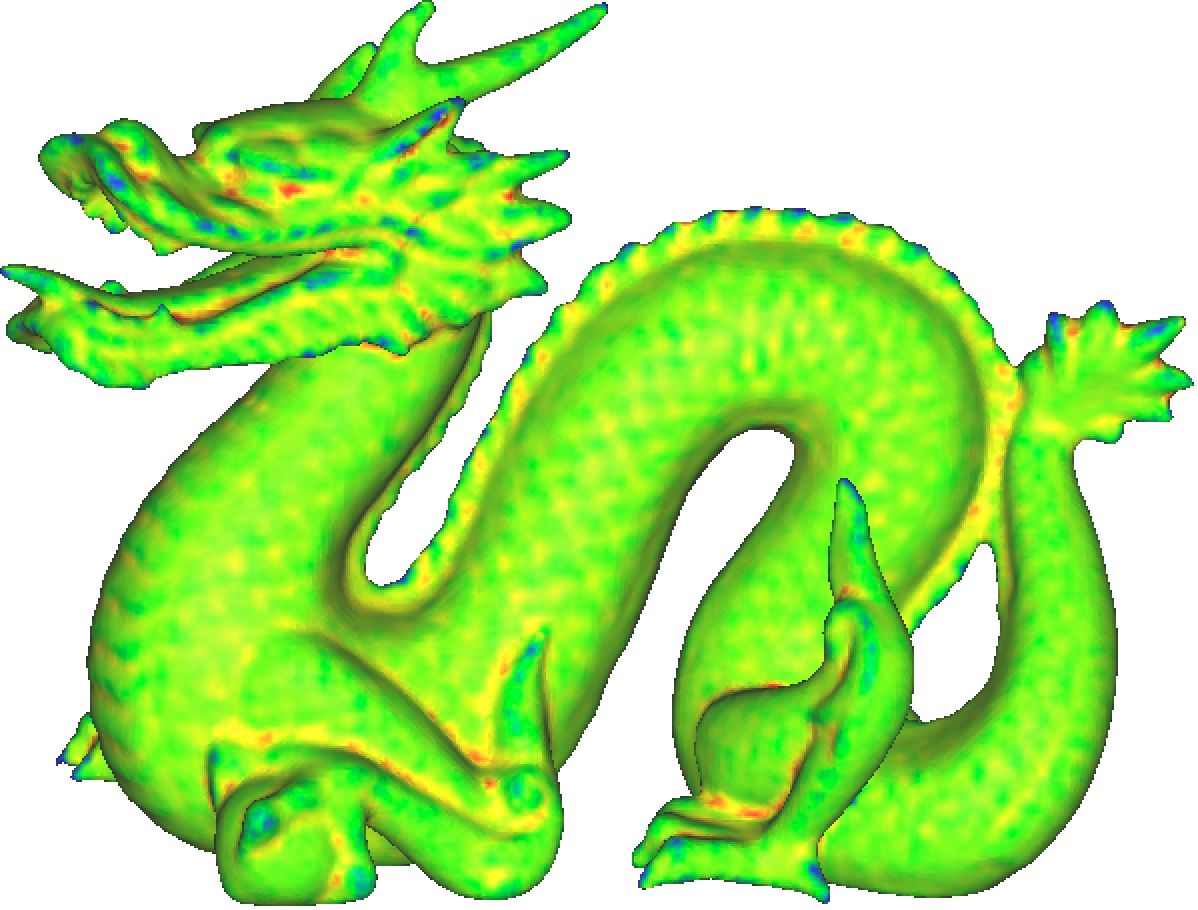}
\label{fig:pcd-on}}
\subfloat[{PDS}]{
\includegraphics[width=0.25\textwidth]{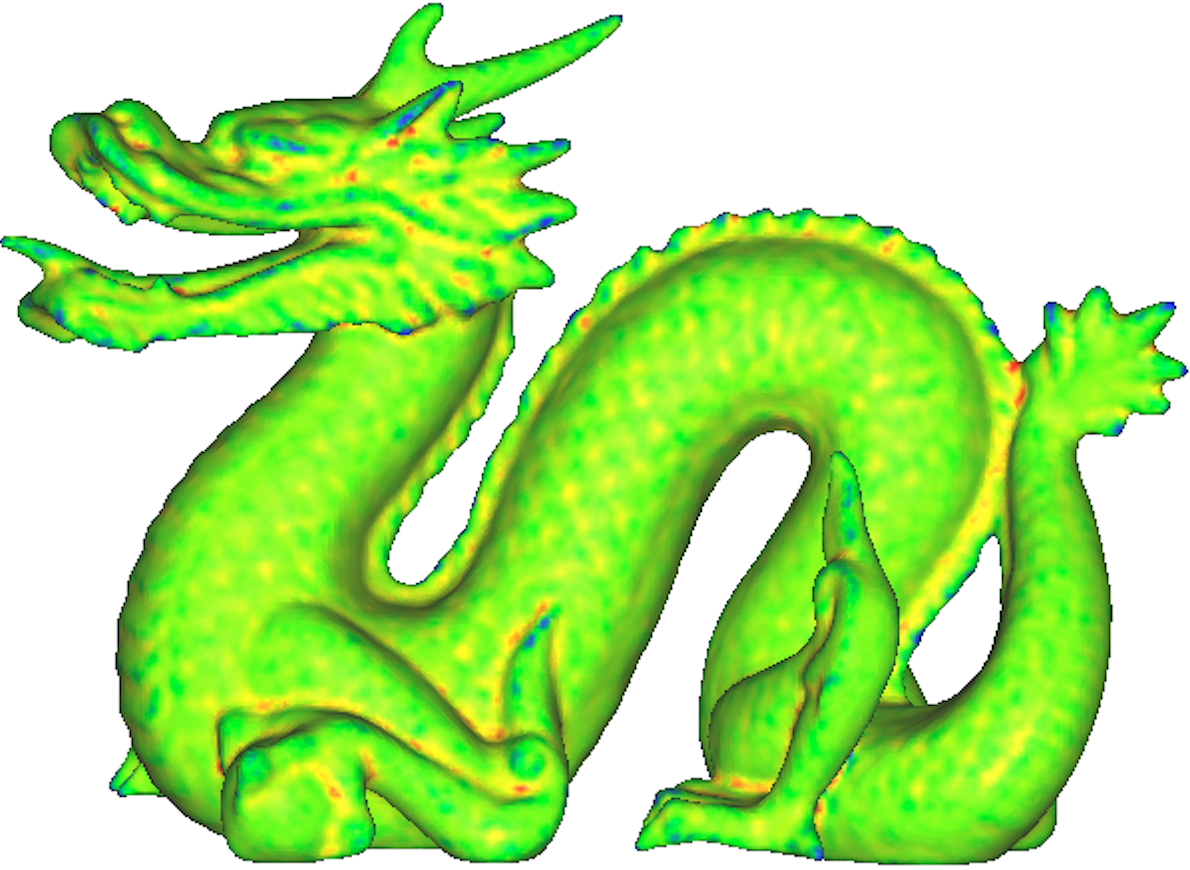}
\label{fig:pcd-off}}
\subfloat[proposed]{
\includegraphics[width=0.25\textwidth]{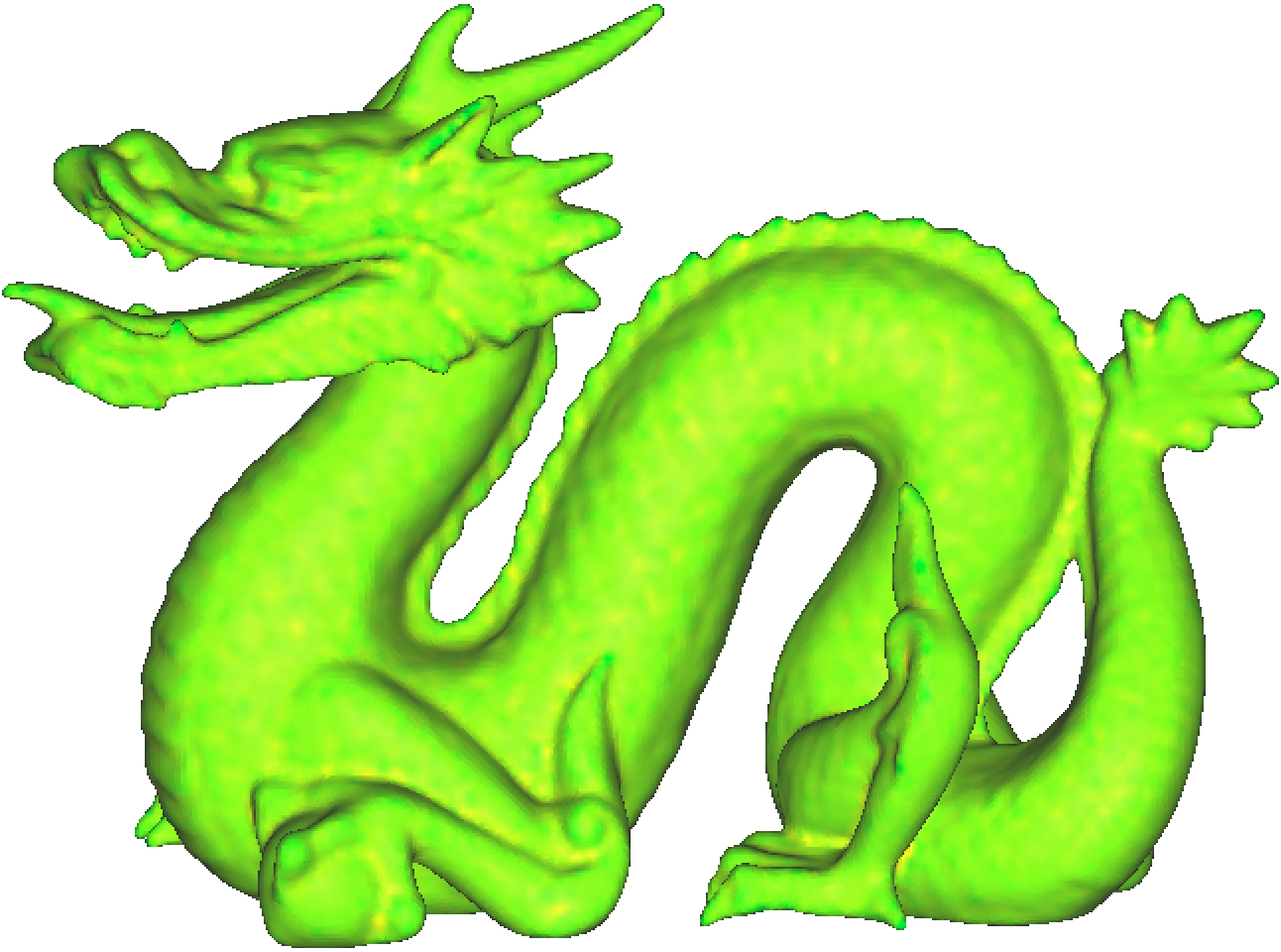}
\label{fig:tvlcd-ofn}}
\subfloat{
\includegraphics[width=0.032\textwidth]{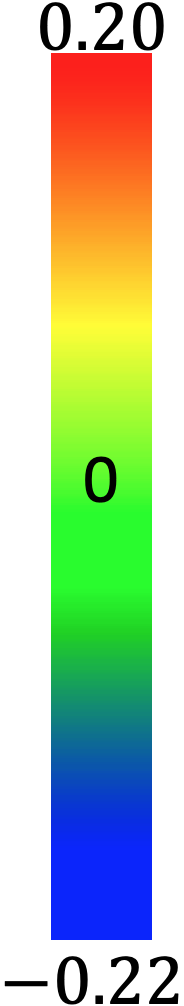}
\label{fig:tvlcd-ofn}}
\caption{SR reconstruction results obtained using EAR SR method from different methods of sub-sampled \texttt{Dragon} models under 0.3 sub-sampling ratio. Here the surfaces are colorized by the distance from the ground truth surface (color-map is included)} 
\label{fig:dragon}
\vspace{-10pt}
\end{figure*}

\begin{table}[t]
\centering
\caption{Comparisons between mean $\lambda_{\min}(\B)\left(\times 10^{-3}\right)$ via AGBS~\cite{dinesh2020} method and mean $\lambda_{\min}(\B)\left(\times 10^{-3}\right)$ via proposed method}
\vspace{-8pt}
\begin{tabular}{c|c|c|c|c|c|}
\cline{2-6}
\multirow{2}{*}{}              & \multicolumn{5}{c|}{Sub-sampling ratio (\%):} \\ \cline{2-6} 
                               & 20     & 30     & 40     & 50     & 60     \\ \hline
\multicolumn{1}{|c|}{AGBS} & 1.92  & 3.84  & 8.53  & 14.5  & 26.8  \\ \hline
\multicolumn{1}{|c|}{Proposed} & \textbf{3.15}  & \textbf{5.05}  & \textbf{12.2}  & \textbf{24.8}  & \textbf{38.4}  \\ \hline
\end{tabular}
\label{tab:mineigen}
\vspace{-10pt}
\end{table}

\begin{table}[t]
\centering
\caption{RE and DCS of balanced graphs obtained from AGBS~\cite{dinesh2020} method and the proposed method for different PC models}
\vspace{-7pt}
\begin{tabular}{|c|c|c|c|c|}
\hline
\multirow{2}{*}{PC models} & \multicolumn{2}{c|}{RE} & \multicolumn{2}{c|}{DCS} \\ \cline{2-5} 
                           & AGBS   & Proposed   & AGBS    & Proposed   \\ \hline
\texttt{Bunny}                      & 0.357      & \textbf{0.302}      & 0.618       & \textbf{0.670}      \\ \hline
\texttt{Dragon}                     & 0.441      & \textbf{0.371}      & 0.634       & \textbf{0.678}      \\ \hline
\texttt{Armadillo}                  & 0.452      & \textbf{0.385}      & 0.597       & \textbf{0.635}      \\ \hline
\texttt{Happy Buddha}               & 0.371      & \textbf{0.338}      & 0.628       & \textbf{0.687}      \\ \hline
\texttt{Asian Dragon}               & 0.411      & \textbf{0.357}      & 0.612       & \textbf{0.683}      \\ \hline
\texttt{Lucy}                       & 0.503      & \textbf{0.434}      & 0.504       & \textbf{0.553}      \\ \hline
\texttt{Gargoyle}                      & 0.397      & \textbf{0.348}      & 0.625       & \textbf{0.690}      \\ \hline
\texttt{DC}                     & 0.512      & \textbf{0.465}      & 0.523       & \textbf{0.592}      \\ \hline
\texttt{Daratech}                  & 0.368      & \textbf{0.324}      & 0.635       & \textbf{0.697}      \\ \hline
\texttt{Anchor}               & 0.385      & \textbf{0.334}      & 0.617       & \textbf{0.663}      \\ \hline
\texttt{Lordquas}               & 0.437      & \textbf{0.388}      & 0.612       & \textbf{0.670}      \\ \hline
\texttt{Julius}                       & 0.507      & \textbf{0.448}      & 0.493       & \textbf{0.538}      \\ \hline
\texttt{Nicolo}               & 0.485      & \textbf{0.430}      & 0.631       & \textbf{0.676}      \\ \hline
\texttt{Laurana}                       & 0.527      & \textbf{0.481}      & 0.512       & \textbf{0.557}      \\ \hline
\texttt{redandblack}    & 0.541      & \textbf{0.494}      & 0.521       & \textbf{0.582}      \\ \hline
\texttt{longdress}                     & 0.512      & \textbf{0.471}      & 0.523       & \textbf{0.561}      \\ \hline
\end{tabular}
\label{tab:RE_DCS}
\vspace{-10pt}
\end{table}

We present comprehensive experiments to verify the effectiveness of our proposed PC sub-sampling strategy. 
Specifically, we compared the performance of our scheme against 10 existing PC sub-sampling methods (we used at least one method per category of existing PC sampling works discussed in Section~\ref{sec:related}): STS~\cite{Ranzuglia2012}, MCS~\cite{Ranzuglia2012}, MESS~\cite{Ranzuglia2012}, PDS~\cite{corsini2012efficient},  random sub-sampling (RS)~\cite{pomerleau2013comparing}, box grid filter-based sub-sampling (BGFS)~\cite{rusu2012downsampling}, clustering-based sub-sampling (CS)~\cite{shi2011}, progression sub-sampling (PS)~\cite{yang2015feature}, feature preserved sub-sampling (FPS)~\cite{qi2019}, and our previous work on ad-hoc graph balancing based sub-sampling (AGBS)~\cite{dinesh2020ICIP}. For the experiments, we used 16 high-resolution PC models\footnote{\texttt{Bunny}, \texttt{Dragon}, \texttt{Armadillo}, \texttt{Happy Buddha}, \texttt{Asian Dragon}, \texttt{Lucy} from~\cite{levoy2005stanford}, \texttt{Gargoyle}, \texttt{DC}, \texttt{Daratech}, \texttt{Anchor}, \texttt{Lordquas} from~\cite{berger2013}, \texttt{Julius}, \texttt{Nicolo} from~\cite{zhang2015}, \texttt{Laurana} from~\cite{potenziani2015}, and \texttt{redandblack}, \texttt{longdress} from~\cite{d20168i}}. 
All PC models were first rescaled, so that each tightly fitted inside a bounding box with the same diagonal. 
During the experiments, PC models were sub-sampled to 20\%, 30\%, 40\%, 50\%, and 60\% of the original points using our scheme and existing methods.

Our PC sub-sampling method can be slower for larger-sized PCs due to computational complexity of GDAS~\cite{yuanchao2019ICASSP,bai2020} for very large graphs. 
Thus, we first divided a PC into $c$ clusters (or \textit{sub-clouds}), so that each sub-cloud contained $M$ points on average (in our experiments, we use $M~=~10000$).
We then performed sub-sampling for each sub-cloud independently. 
For this purpose, we used $K$-means clustering~\cite{xu2018}, using 3D coordinates as the feature for clustering \cite{xu2018,dinesh2019}. An example of such sub-clouds obtained using the \texttt{Bunny} PC model~\cite{levoy2005} is shown in Fig.\,\ref{fig:clust}. 
Here, the model is divided into four sub-clouds, and each color represents a different sub-cloud.

\subsection{Reconstruction Error Comparisons}
\label{sec:rec_error}

To demonstrate the performance of our PC sub-sampling scheme, we compared the SR reconstruction errors of our method against existing PC sub-sampling methods. 
Specifically, sub-sampled PCs were super-resolved (reconstructed) to approximately same number of points as the ground truths using different SR methods in the literature---APSS~\cite{guennebaud2008}, EAR~\cite{huang2013}, and GTV~\cite{dinesh2019}, and FGTV~\cite{dinesh2020ICASSP}. 
Then, for quantitative SR reconstruction error comparisons, we measured the point-to-point (denoted as C2C) error and point to plane (denoted as C2P) error \cite{tian2017} between ground truth and super-resolved PCs. 

For the C2C error, we first measured the average of the Euclidean distances between ground truth points and their closest points in the super-resolved cloud, and also that between the points of the super-resolved cloud and their closest ground truth points. 
Then the larger of these two values was taken as the C2C error. 
For the C2P error, we first measured the average of the Euclidean distances between ground truth points and tangent planes at their closest points in the super-resolved cloud, and also that between the points of the super-resolved cloud and tangent planes at their closest ground truth points. 
Again, the larger of these two values was taken as the C2P error.


Tables \ref{tab:C2C} and \ref{tab:C2P} show the average C2C and C2P errors, respectively, per PC model for sub-sampling ratios\footnote{$\text{sub-sampling ratio}=\frac{\# \text{of points in sub-sampled PC}}{\# \text{of points in ground truth PC}}$} of 0.2, 0.3, 0.4, 0.5, 0.6 under five different SR methods. 
We observe that our scheme achieved the best results on average in terms of C2C and C2P errors for all SR methods with different sub-sampling ratios. 
As an example, for FGTV SR method, our scheme reduced the second best average C2C and C2P among existing sub-sampling schemes by 11.6\%, 7.2\%, 9.5\%, 10.4\%, 11.2\% and 17.5\%, 12.1\%, 12.9\%, 7.5\%, 8.3\%, respectively, for sub-sampling ratio 0.2, 0.3, 0.4, 0.5, 0.6.     

Besides numerical comparison, representative visual results for three models are shown in Fig.\;\ref{fig:bunny}, \ref{fig:armadillo}, and \ref{fig:dragon} where the SR reconstructions of the sub-sampled PCs were colorized by the distance from the ground truth PC to the surface of the super-resolved PC. 
We used CloudCompare\footnote{https://www.danielgm.net/cc/} tool to colorize the super-resolved PC. Here, green and yellow represent the smaller negative and positive errors (distance), while blue and red represent the larger negative and positive errors (the corresponding color-map is included at the right of each figure). 
As discussed in Section\;\ref{sec:related}, existing PC sub-sampling methods in the literature either employ approaches that do not preserve geometric characteristics (like sharp edges, corners, valleys) pro-actively, or maintain those geometric features without ensuring the overall PC-SR reconstruction quality. 
Thus, from Fig.\;\ref{fig:bunny}, \ref{fig:armadillo}, and \ref{fig:dragon}, we observe that the super-resolved models of sub-sampled PC obtained from different methods have distorted surfaces (\textit{i.e.}, red and blue colors) in some valleys or edges. 
In contrast, in our method, we performed PC sub-sampling while systematically minimizing a global PC-SR reconstruction error metric. 
As observed, SR reconstruction results of sub-sampled PC using our scheme have well-preserved geometric details compared to existing schemes.

\subsection{Minimum Eigenvalue Comparisons}
\label{sec:eig_comp}

As we discussed in Section~\ref{sec:sampling_obj}, our PC sub-sampling objective is to maximize $\lambda_{\min}(\B)$. We computed $\lambda_{\min}(\B)$ for each PC model for different sub-sampling ratios given sampling matrices $\H$ obtained via our optimization. 
For comparisons, we computed $\lambda_{\min}(\B)$ given $\H$ obtained via the sub-sampling method in AGBS~\cite{dinesh2020ICIP}. 
Mean $\lambda_{\min}(\B)$ per PC model obtained via both sub-sampling methods is shown in Table~\ref{tab:mineigen} as a function of the sub-sampling ratio. 
As shown in Table~\ref{tab:mineigen}, $\lambda_{\min}(\B)$ for our optimization is significantly larger than $\lambda_{\min}(\B)$ for AGBS, for different sub-sampling ratios.

\subsection{Graph Similarity Comparisons}

A balanced graph $\cG_{B}$ obtained from our graph balancing algorithm is closer to the original graph $\cG$ than $\cG_{B}$ obtained from AGBS \cite{dinesh2020ICIP}. 
To effectively measure similarity between two graphs, we use two different metrics called \textit{relative error} (RE)~\cite{egilmez2017} and \textit{DELTACON similarity} (DCS) ~\cite{koutra2013}. 
RE between $\cG$ and $\cG_{B}$ is computed~\cite{egilmez2017} as
\begin{equation}
    \text{RE}=\frac{\norm{\L-\L_{B}}_{F}}{\norm{\L}_{F}},
\end{equation}
where $\norm{\cdot}_{F}$ is the Frobenius norm. Metric DCS computes the pairwise node similarities in $\cG$, and compares them with the ones in $\cG_{B}$ as follows (see~\cite{koutra2013} for more information).     

First a similarity matrix $\S^{\mathcal{X}}$ is computed for graph $\mathcal{X}$, where $\mathcal{X}$ is either original graph $\cG$ or its approximated balance graph $\cG_{B}$:
\begin{equation}
    \S^{\mathcal{X}}=[\I+\epsilon^{2}\D_{\mathcal{X}}-\epsilon\A_{\mathcal{X}}]^{-1},
\end{equation}
where $\A_{\mathcal{X}}$ and $\D_{\mathcal{X}}$ are adjacency matrix and degree matrix of graph $\mathcal{X}$. 
$\epsilon>0$ is a small constant. 
Then Matusita distance $d_{m}$ between $\cG$ and $\cG_{B}$ is calculated as
\begin{equation}
    d_{m}=\sqrt{\sum_{i=1}^{3n}\sum_{j=1}^{3n}\left(\sqrt{s^{\cG}_{i,j}}-\sqrt{s^{\cG_{B}}_{i,j}}\right)^{2}}.
\end{equation}
Finally, DCS metric is computed as $\frac{1}{1+d_{m}}$; it takes values between 0 and 1, where the value 1 means the perfect similarity.

As shown in Table~\ref{tab:RE_DCS}, a balanced graph $\cG_{B}$ computed using our graph balancing algorithm is closer to the original graph $\cG$ than $\cG_{B}$ obtained from AGBS \cite{dinesh2020ICIP}, in terms of both RE and DCS metrics.

\section{Conclusion}
\label{sec:conclude}
Acquired point clouds tend to be very large in size, which leads to high processing costs.
We proposed a fast point cloud sub-sampling method based on graph sampling, 
where the objective is to minimize the worst-case super-resolution (SR) reconstruction error via selection of a sampling matrix $\H$. 
We showed that this is equivalent to maximizing the smallest eigenvalue $\lambda_{\min}$ of a coefficient matrix $\H^{\top} \H + \mu \cL$, where $\cL$ is a symmetric positive semi-definite (PSD) matrix computed from connectivity information of neighboring 3D points.
We proposed to maximize the lower bound $\lambda^-_{\min}(\H^{\top} \H + \mu \cL)$ efficiently instead in three steps.
Interpreting $\cL$ as a generalized graph Laplacian matrix of an unbalanced signed graph $\cG$, we first approximated $\cG$ with a balanced graph $\cG_B$, resulting in generalized Laplacian $\cL_B$. 
Leveraging on a recent theorem Gershgorin disc perfect alignment (GDPA), we then performed a similarity transform $\cL_p = \S \cL_B \S^{-1}$ where $\S = \text{diag}(v_1, v_2, \ldots)$ and $\v$ is the first eigenvector of $\cL_B$, so that $\cL_p$ has all its Gershgorin disc left-ends aligned at the same value $\lambda_{\min}(\cL_B)$.
Finally, we applied an eigen-decomposition-free graph sampling algorithm GDAS in roughly linear time to choose 3D points to maximize $\lambda^-_{\min}(\H^{\top} \H + \mu \cL_p)$.
Experimental results showed that our PC sub-sampling scheme outperforms competitors, both numerically and visually, in SR reconstruction quality. 


\ifCLASSOPTIONcaptionsoff
  \newpage
\fi


\bibliographystyle{IEEEtran}
\vspace{-0pt}
\bibliography{refs}
\vspace{-30pt}

\end{document}